\documentclass{article}

\usepackage{arxiv}

\usepackage[utf8]{inputenc} 
\usepackage[T1]{fontenc}    
\usepackage{hyperref}       
\usepackage{url}            
\usepackage{booktabs}       
\usepackage{amsfonts}       
\usepackage{nicefrac}       
\usepackage{microtype}      
\usepackage{graphicx}
\usepackage[numbers, compress]{natbib}
\usepackage{doi}

\usepackage{orcidlink}
\usepackage{algorithm}
\usepackage{algpseudocode}
\usepackage{amsmath, amssymb, amsthm}
\usepackage{subcaption}
\usepackage[normalem]{ulem} 
\usepackage{multirow}
\usepackage{xcolor}         
\usepackage{makecell}
\usepackage{empheq}
\usepackage[most]{tcolorbox}
\usepackage{wrapfig}

\algnewcommand\Input{\item[\textbf{Input:}]}
\algnewcommand\Output{\item[\textbf{Output:}]}

\theoremstyle{plain}
\newtheorem{theorem}{Theorem}[section]
\newtheorem{proposition}[theorem]{Proposition}
\newtheorem{lemma}[theorem]{Lemma}
\newtheorem{corollary}[theorem]{Corollary}
\theoremstyle{definition}

\newtheorem{assumption}[theorem]{Assumption}

\theoremstyle{remark}
\newtheorem{remark}[theorem]{Remark}

\usepackage{orcidlink}
 \usepackage{algorithm}
 \usepackage{algpseudocode}
 \usepackage{amsmath, amssymb, amsthm}

\theoremstyle{plain}

\usepackage[utf8]{inputenc} 
\usepackage[T1]{fontenc}    
\usepackage{hyperref}       
\usepackage{url}            
\usepackage{booktabs}       
\usepackage{amsfonts}       
\usepackage{nicefrac}       
\usepackage{microtype}      
\usepackage{xcolor}         

\usepackage[utf8]{inputenc} 
\usepackage[T1]{fontenc}    
\usepackage{hyperref}       
\usepackage{url}            
\usepackage{booktabs}       
\usepackage{amsfonts}       
\usepackage{nicefrac}       
\usepackage{microtype}      
\usepackage{xcolor}         
\usepackage{empheq}
\usepackage[most]{tcolorbox}
\usepackage{wrapfig}

\newtcbox{\eqbox}{
  colback=blue!15,
  colframe=black,
  arc=3mm,
  boxrule=0.8pt,
  left=6pt,
  right=6pt,
  top=6pt,
  bottom=6pt
}

\title{OASIS: Observation-Aware Simulation-Based Inference via Distributional Matching}

\author{%
  Arya~Farahi \orcidlink{0000-0003-0777-4618}\\
  Department of Statistics and Data Sciences, University of Texas at Austin, Austin, Texas 78712, USA\\
  The NSF-Simons AI Institute for Cosmic Origins, USA \\
  \texttt{arya.farahi@austin.utexas.edu}\\
    \AND
  Conghao~Zhou \orcidlink{tbd}\\
  Department of Physics, University of California, Santa Cruz, CA 95064, USA\\
  \texttt{zhou.conghao@ucsc.edu}\\
    \AND
  Ritwik~Vashistha \orcidlink{tbd}\\
  Department of Statistics and Data Sciences, University of Texas at Austin, Austin, Texas 78712, USA\\
  The NSF-Simons AI Institute for Cosmic Origins, USA \\
    \texttt{ritwik.v@utexas.edu} 
}

\date{}

\begin{document}

\maketitle

\begin{abstract}
  We introduce \texttt{OASIS}, a simulation-based inference framework for scientific settings where observations are distorted by measurement error, selection effects, and other survey-specific transformations. In many real applications, simulators generate latent, noiseless quantities, while the data are observed only after passing through a complex observational pipeline. Standard simulation-based inference methods often ignore this distinction, comparing observations to idealized simulator outputs or relying on low-dimensional summaries that can miss important structure. \texttt{OASIS} addresses this mismatch by explicitly embedding the observation model into the simulator and performing inference directly at the level of observed-data distributions. The method constructs a pseudo-posterior by reweighting prior samples according to a maximum mean discrepancy (MMD) loss between the empirical distributions of the observed data and forward-simulated observations, thereby avoiding both handcrafted summaries and learned neural surrogates. We provide theoretical guarantees for Monte Carlo consistency, convergence of the empirical pseudo-posterior to its population counterpart, and posterior concentration on the MMD-identified parameter set, with consistency for the true parameter under correct specification and identifiability. In controlled errors-in-variables regression experiments, \texttt{OASIS} delivers robust parameter recovery and well-calibrated uncertainty under heterogeneous and non-Gaussian measurement noise. We then demonstrate the method on a realistic cosmological application involving galaxy cluster observations across multiple wavelengths, in which latent physical properties are linked to observables through nonlinear scaling relations, heteroscedastic errors, selection functions, and incomplete coverage.
\end{abstract}

\section{Introduction}

Modern scientific inference is increasingly driven by complex simulators that generate high-dimensional data under richly structured physical models \citep{karniadakis2021physics}. In many domains, such as cosmology, particle physics, climate science, and systems biology, the data-generating process is only partially observable. Meaning latent quantities of interest are transformed by measurement noise, selection effects, censoring, and survey-specific distortions before they are recorded \citep{carroll2006measurement,fuller2009measurement}. As a result, inference must proceed by comparing observed data not to idealized simulator outputs, but to their \emph{observationally transformed} counterparts. This distinction is central in practice, yet remains under-addressed in much of the simulation-based inference (SBI) literature \citep{cranmer2020frontier}, where comparisons are often performed either on latent quantities or on low-dimensional summary statistics that do not faithfully reflect the full observational pipeline.

Approximate Bayesian computation (ABC) and related likelihood-free methods provide a flexible framework for SBI when the likelihood is intractable, but forward simulation is available \citep{beaumont2002approximate,sisson2018handbook}. However, classical ABC relies on carefully chosen summary statistics and distance functions, which can lead to information loss and bias, particularly in high-dimensional settings or when observational distortions are complex. Recent approaches have sought to mitigate these issues using learned summaries \citep{papamakarios2019sequential} or distributional distances such as the maximum mean discrepancy (MMD) \citep{fukumizu2011kernel,park2016k2}, which operate directly on empirical distributions. While these advances improve expressivity, they typically assume that simulated and observed data are directly comparable, implicitly neglecting the structured observation processes that intervene between latent states and recorded measurements.

Neural simulation-based inference methods (e.g., neural posterior, likelihood, and ratio estimation) approximate $p(\theta\!\!\mid\!\!y)$, $p(y\!\!\mid\!\!\theta)$, or density ratios using learned models trained on simulated data \citep{cranmer2015approximating,papamakarios2016fast,lueckmann2019likelihood,hermans2020likelihood}. These approaches scale to high-dimensional observations and enable amortized inference \citep{zammit2025neural}, but require demanding computational resources, careful tuning, and large simulation budgets, and can be sensitive to architectural and optimization choices. More fundamentally, they assume that training simulations match the observed-data distribution. In many scientific settings, simulators produce latent quantities that must be transformed through complex and imperfectly characterized observational pipelines. If this mismatch is not explicitly modeled, the learned representations may be misaligned with the observed data, leading to biased or miscalibrated inference \citep{hermans2022crisis}. This motivates inference procedures that operate directly on observed-data distributions while explicitly incorporating observational mechanisms, instead of relying on surrogates trained on idealized simulations.

We introduce \texttt{OASIS}\footnote{OASIS:Observation-Aware Simulation Inference Scheme}, an SBI framework that incorporates observational mechanisms and operates directly on observed-data distributions. The method constructs a pseudo-posterior by comparing empirical observed data with simulated data after passing through an observation simulator with MMD \citep{gretton2006kernel,fukumizu2011kernel,vashistha2026convolutional}. This yields a likelihood-free approach that is distributionally expressive, avoids handcrafted summaries, and aligns inference with the data acquisition process. \texttt{OASIS} has a Gibbs-type pseudo-posterior interpretation defined through a distributional loss \citep{geman1984stochastic,jiang2008gibbs,bissiri2016general}, with two key differences: the loss is stochastic due to Monte Carlo simulation, inducing dependence on the simulation budget and a temperature parameter \citep{farahi2026simulation}, and the observational model is embedded in the forward simulation, so inference targets the observed-data law. We establish Monte Carlo consistency, convergence of the empirical pseudo-posterior, and posterior concentration on minimizers of the MMD discrepancy, yielding consistency under correct specification and identifiability despite measurement error and selection effects. We benchmark against error-in-variables methods in linear regression and evaluate on an astrophysical inference problem involving galaxy clusters \citep{allen2011cosmological}, where nonlinear relations, heteroscedastic noise, selection effects, and partial coverage challenge standard approaches. By forward simulating the full observational pipeline and comparing distributions, \texttt{OASIS} achieves robust parameter recovery and well-calibrated uncertainty.

\textbf{Contributions.}
We advance SBI by incorporating observational mechanisms into the inference loop and providing a theoretically grounded, distribution-level alternative to summary-based methods, combining kernel techniques, simulation-based modeling, and generalized Bayesian inference. \vspace{-2mm}

\section{\texttt{OASIS} Method}
\label{sec:setup}

\paragraph{Context.}
\texttt{OASIS} targets settings where inference must account for measurement error and observational distortions. In many applications, the simulator generates a latent, noiseless population, whereas the observed dataset is obtained only after noise, heteroscedastic scatter, missingness, censoring, truncation, thresholding, or survey selection have acted on that population. Consequently, the relevant comparison is not between observed data and latent simulator output, but between observed data and simulator output after the same observational mechanism has been applied.

\paragraph{Notation.}
We observe $\mathcal D_{\mathrm{obs}}=\{y_i^{\mathrm{obs}}\}_{i=1}^n$, with $y_i^{\mathrm{obs}}\in\mathbb R^d$, drawn from an observed-data distribution $P_{\mathrm{obs}}$. Each $y_i^{\mathrm{obs}}$ represents the output of the observational pipeline. The goal is to compare the empirical distribution of $\mathcal D_{\mathrm{obs}}$ with that induced by a simulator under parameter $\theta\in\Theta\subseteq\mathbb R^p$, with prior $p(\theta)$. For each $\theta$, the simulator generates latent variables $u^{\mathrm{true}}\sim P_\theta^{\mathrm{true}}$, followed by an observation step $y^{\mathrm{sim,obs}}\sim P^{\rm Err}_\theta(\cdot\mid u^{\mathrm{true}})$, where $P^{\rm Err}_\theta$ encodes measurement error, noise, and selection effects. We assume samples from $P^{\rm Err}_\theta$ can be simulated. While strong, it holds in many practical settings \citep[e.g.,][]{wimmer2000proper,creevey2013large,chen2019calibration,leung2019simultaneous,kuhn2019kinematics,glazer2026beyond}. The induced observed-data law is
\begin{equation} \label{eq:obs-law}
P_\theta^{\mathrm{obs}}(Y^{\rm sim}) = \int P^{\rm Err}_\theta(Y^{\rm sim}\mid u,\phi)\,P_\theta^{\mathrm{true}}(du),
\qquad
Y^{\rm sim}\subseteq \mathbb R^d.
\end{equation}
Inference targets values of $\theta$ for which $P_\theta^{\mathrm{obs}}$ is close to $P_{\mathrm{obs}}$. In our method, a closed-form solution for $P^{\rm Err}$ is not a requirement; we only need to be able to sample from it.

\paragraph{Inference Method (\texttt{OASIS}).}
The method proceeds by comparing the empirical observed-data distribution with forward simulations that explicitly incorporate both the latent data-generating process and the observational mechanism. See Algorithm~\ref{alg:OASIS} in Supplementary Materials.

We draw $\theta_1,\dots,\theta_{n_\theta} \sim p(\theta)$.  For each parameter value $\theta$, we generate a Monte Carlo approximation of the observed-data law $P_\theta^{\mathrm{obs}}$ through a two-stage procedure. Then, we draw latent samples
\begin{equation}
u_{\theta,m}^{\mathrm{true}} \sim P_\theta^{\mathrm{true}}, \qquad m=1,\dots,M.
\end{equation}
Second, each latent realization is mapped through the observation model by sampling
\begin{equation}
y_{\theta,m}^{\mathrm{sim,obs}} \sim P^{\rm Err}_\theta(\cdot \mid u_{\theta,m}^{\mathrm{true}}).
\end{equation}
This yields a simulated observed dataset $\{y_{\theta,m}^{\mathrm{sim,obs}}\}_{m=1}^M$, with associated empirical measure $\widehat P_{\theta,M}^{\mathrm{sim,obs}} = \frac{1}{M}\sum_{m=1}^M \delta_{y_{\theta,m}^{\mathrm{sim,obs}}}$, which represents the distribution of observations induced by $\theta$ after passing through the full observational pipeline.

We assess agreement between the observed and simulated data at the distributional level using MMD \citep{gretton2006kernel,gretton2012kernel}. For a bounded positive-definite kernel $k$ with RKHS $\mathcal H_k$,
\begin{equation}
\mathrm{MMD}_k^2(P_{\mathrm{obs}}, P_{\theta,M}^{\mathrm{sim,obs}}) = \|\mu_k(P_{\mathrm{obs}})-\mu_k(P_{\theta,M}^{\mathrm{sim,obs}})\|_{\mathcal H_k}^2,
\label{eq:mmd-def}
\end{equation}
where $\mu_k(P)=\int k(y,\cdot)\,P(dy)$. If $k$ is characteristic, $\mathrm{MMD}_k=0$ implies $P_{\mathrm{obs}}=P_{\theta,M}^{\mathrm{sim,obs}}$.  \cite{vashistha2026convolutional} provide identifiability conditions for the observation model under which equality of observed-data laws implies $P_{\mathrm{true}} = P_\theta^{\mathrm{true}}$. The empirical loss, defined as
\begin{equation}
\Delta_{n,M}(\theta) := \mathrm{MMD}_k^2\!\bigl(\widehat P_{\mathrm{obs}}, \widehat P_{\theta,M}^{\mathrm{sim,obs}}\bigr),
\end{equation}
is small when the forward-simulated observed distribution matches the empirical distribution of data.

Inference is performed by constructing a likelihood-free pseudo-posterior through weighted prior sampling. We assign weights
\begin{equation}
\widetilde w_j = \exp\!\left\{-\frac{\Delta_{n,M}(\theta_j)}{2\tau^2}\right\}, 
\qquad
w_j = \frac{\widetilde w_j}{\sum_{k=1}^{n_\theta}\widetilde w_k},
\label{eq:weights}
\end{equation}
yielding the weighted empirical measure $\Pi_{n_\theta}^{\mathrm{MMD}} = \sum_{j=1}^{n_\theta} w_j\,\delta_{\theta_j}$. Parameters has high weight when their simulated observed-data distributions are close to the empirical distribution of the observations.
At the population level, the method targets the pseudo-posterior
\begin{equation}
\pi_\tau(d\theta)\propto p(\theta)\exp\{-\Delta(\theta)/(2\tau^2)\}d\theta,
\qquad
\Delta(\theta)=\mathrm{MMD}_k^2(P_{\mathrm{obs}},P_\theta^{\mathrm{obs}}).
\end{equation}
As $\tau \downarrow 0$, $\pi_\tau$ concentrates on $\Theta^\dagger = \arg\min_{\theta\in\Theta} \Delta(\theta)$, the set of parameters whose induced observed-data laws match $P_{\mathrm{obs}}$. Under correct specification and identifiability, this set reduces to the true parameter $\theta^\star$. The weighted sample $\{(\theta_j,w_j)\}$ serves as a posterior surrogate. For any test function $h$, $\int h(\theta)\,d\Pi^{\mathrm{MMD}} \approx \sum_j w_j h(\theta_j)$, with posterior predictive quantities obtained via forward simulation.

\paragraph{Relation to Literature.} 
The proposed framework is related to several strands of likelihood-free and generalized Bayesian inference \citep{jiang2008gibbs,bissiri2016general,cranmer2020frontier}, but differs from existing formulations in the object being matched and in the role of the observational mechanism. From the perspective of generalized Bayes, our update has the form of a Gibbs posterior $\pi(\mathrm d\theta\mid y_{\mathrm obs})\propto p(\theta)\exp\{-\lambda \ell(y_{\mathrm obs},\theta)\}\mathrm d\theta$, with the loss defined by a discrepancy between the empirical observed-data distribution and the simulated observed-data distribution. Unlike the standard setup in generalized Bayes, however, the loss is not deterministic conditional on the data. It is itself induced by Monte Carlo draws from the simulator, which is composed of the observation model. This yields an $(M,\tau)$-indexed pseudo-posterior in which $M$ controls Monte Carlo variability and $\tau$ controls the sharpness of the update. In the large-$M$ regime, the construction approaches a deterministic Gibbs posterior based on the population discrepancy $\Delta(\theta)$, linking the method to generalized Bayes, Gibbs posteriors, and minimum-distance or risk-based inference with well-characterized asymptotics \citep{miller2021asymptotic}. This also distinguishes the present construction from robust generalized Bayesian procedures such as coarsened or calibrated posteriors, where robustness is introduced through alternative update rules or temperature calibration rather than through an explicit forward model for the data-acquisition process \citep{miller2019robust,holmes2017assigning,martin2022direct}. 

The method is also closely connected to approximate Bayesian computation \citep{beaumont2002approximate,frazier2018asymptotic,sisson2018handbook}. Classical ABC accepts or reweights parameter draws according to the proximity of simulated and observed data, typically through low-dimensional summaries, while semi-automatic ABC and related variants learn or construct summaries designed to improve efficiency. Synthetic-likelihood methods replace the intractable likelihood by a parametric model, often Gaussian, for summary statistics, thereby reducing dimensionality at the cost of an additional approximation \citep{wood2010statistical,price2018bayesian}. Indirect inference and auxiliary-model ABC make a similar tradeoff by comparing observed and simulated data through auxiliary estimators, scores, or estimating equations \citep{drovandi2015bayesian,drovandi2018abc}. In contrast, our approach operates directly on empirical measures and does not require either handcrafted summaries or an auxiliary parametric likelihood. The price is that the comparison takes place in a richer function space, shifting the burden from summary design to the choice of kernel and the simulation budget. 

Among distributional ABC methods, the closest antecedents are kernel ABC \citep{fukumizu2011kernel} and K2-ABC \citep{park2016k2}, which use kernel mean embeddings and MMD to compare empirical distributions. At the algorithmic level, our weighting rule is indistinguishable from K2-ABC. The key conceptual difference is that our method is derived from generalized Bayes principles, similar to MMD-Bayes \citep{cherief2020mmd}. The primary technical difference is that we make the observation operator explicit and define the target of inference through the induced observed-data law $P_\theta^{\mathrm{obs}}$, rather than treating simulator output and observations as directly comparable objects. This distinction is required in scientific settings with measurement error, heteroscedastic noise, censoring, truncation, missingness, or selection effects, where the latent simulator law and the observed-data law are fundamentally different. In that sense, the proposed method is not merely a kernelized ABC procedure, but a distributional pseudo-posterior defined on the image of the simulator under a known or quantifiable observational map. 

More broadly, the method belongs to a class of likelihood-free procedures based on discrepancies between probability measures, including Wasserstein ABC \citep{bernton2019approximate}, sliced-Wasserstein ABC \citep{nadjahi2020approximate}, and classifier- or two-sample-test-based simulation-based inference \citep{chatterjee2024kernel}. These methods share the principle of replacing likelihood evaluation by distributional comparison, but differ in geometry, statistical regularity, and computational scaling. Relative to Wasserstein-based approaches, MMD yields a reproducing-kernel representation with well-developed empirical process tools and a direct connection to kernel two-sample testing \citep{gretton2012kernel}. Relative to neural SBI methods such as neural posterior, likelihood, and ratio estimation, our approach forgoes an amortized neural surrogate in favor of direct particle reweighting \citep{papamakarios2019sequential,cranmer2020frontier,hermans2020likelihood}. This avoids approximation error from training a high-capacity density estimator, but it also sacrifices the amortization and sample-efficiency gains that neural methods can provide in repeated-inference settings. The contribution of the present work is therefore the integration of three ingredients that are usually treated separately: an explicit observation model, distribution-level comparison through a characteristic-kernel MMD loss, and a pseudo-posterior formulation with asymptotic guarantees tailored to the observationally realistic setting. \vspace{-2mm}

\section{Theoretical Guarantees}
\label{sec:theory}

This section develops the theoretical properties of the MMD-based pseudo-posterior introduced in Section~\ref{sec:setup}. The central object is the population discrepancy
\begin{equation}
\Delta(\theta)
=
\mathrm{MMD}_k^2\!\bigl(P_{\mathrm{obs}},P_\theta^{\mathrm{obs}}\bigr),
\end{equation}
together with its empirical counterpart
\begin{equation}
\Delta_{n,M}(\theta)
=
\mathrm{MMD}_k^2\!\bigl(\widehat P_{\mathrm{obs}},\widehat P_{\theta,M}^{\mathrm{sim,obs}}\bigr),
\end{equation}
where $\widehat P_{\mathrm{obs}}$ is the empirical distribution of the observed sample and $\widehat P_{\theta,M}^{\mathrm{sim,obs}}$ is the empirical distribution of a simulated observed batch generated under~$\theta$. With $\theta_1,\dots,\theta_{n_\theta}\stackrel{\mathrm{iid}}{\sim}p(\theta)$, recall that the weighted empirical measure is
\begin{equation}
\Pi_{n_\theta}^{\mathrm{MMD}}(d\theta) = \sum_{j=1}^{n_\theta} w_j\,\delta_{\theta_j}(d\theta),
\qquad
w_j = \frac{\exp\{-\Delta_{n,M}(\theta_j)/(2\tau^2)\}} {\sum_{k=1}^{n_\theta}\exp\{-\Delta_{n,M}(\theta_k)/(2\tau^2)\}}.
\end{equation}

This analysis establishes three distinct guarantees.  \emph{(i) Monte Carlo consistency}: finite weighted particles approximate the empirical pseudo-posterior. \emph{(ii) Statistical convergence}: the empirical pseudo-posterior approaches its population analogue as $n,M\to\infty$. \emph{(iii) Posterior consistency}: when $\tau_{n,M}\downarrow0$ at a rate slow enough relative to the empirical approximation error, the pseudo-posterior concentrates on the MMD-identified set, and under injectivity of the observed-data forward map, on the true parameter itself. Hence, the method enjoys a posterior-consistency property relative to the distributional target induced by MMD. See Appendix~\ref{app:assumptions} for the list of assumptions. 

The theoretical results presented here are largely built from standard ingredients, but their combination in the present setting yields a nontrivial synthesis. The Monte Carlo consistency result (Theorem~\ref{thm:mc}) is a direct consequence of classical laws of large numbers for self-normalized importance sampling. The fixed-temperature convergence result (Theorem~\ref{thm:fixedtau}) follows from uniform convergence of the empirical loss together with standard continuous mapping arguments, and is closely related to existing results in empirical process theory and M-estimation \citep{van2000asymptotic}. The identification statements based on characteristic kernels are also standard consequences of MMD  \citep{gretton2012kernel,vashistha2026convolutional}. The main result is the posterior concentration result (Theorem~\ref{thm:consistency}), which combines uniform stochastic control of the MMD loss with a vanishing-temperature Gibbs posterior argument to establish concentration on the MMD-identified set. While concentration arguments have been studied in the generalized Bayesian and Gibbs posterior literature \citep{mcgoff2022gibbs,lopes2022bayes}, the present result adapts these ideas to a simulation-based setting with two sources of randomness (observational and simulation samples) and a distributional loss defined via kernel embeddings. Thus, although each component builds on established results, their integration in SBI setting is not immediate and requires a careful combination of empirical process arguments, kernel-based distributional theory, and Gibbs posterior asymptotics. Particularly, the resulting framework provides a unified treatment of MMD-based SBI with an explicit observational forward model, presence of measurement error, and dual sampling variability, a setting that is not directly covered by existing results in the kernel ABC or generalized Bayesian literature.

\subsection{Convergence and Consistency of Pseudo-posterior}

For fixed $\tau>0$, define the population pseudo-posterior
\begin{equation}
\pi_\tau(d\theta)
=
\frac{p(\theta)\exp\{-\Delta(\theta)/(2\tau^2)\}\,d\theta}
{\int_\Theta p(\vartheta)\exp\{-\Delta(\vartheta)/(2\tau^2)\}\,d\vartheta}.
\label{eq:population-posterior}
\end{equation}
Likewise, define the empirical target
\begin{equation}
\pi_{n,M,\tau}(d\theta)
=
\frac{p(\theta)\exp\{-\Delta_{n,M}(\theta)/(2\tau^2)\}\,d\theta}
{\int_\Theta p(\vartheta)\exp\{-\Delta_{n,M}(\vartheta)/(2\tau^2)\}\,d\vartheta}.
\label{eq:empirical-target}
\end{equation}
The weighted sample $\Pi_{n_\theta}^{\mathrm{MMD}}$ is a self-normalized importance-sampling approximation to $\pi_{n,M,\tau}$.

For a measurable $h:\Theta\to\mathbb R$ integrable with respect to $\pi_{n,M,\tau}$, define
\begin{equation}
\Phi_{n,M,\tau}(h) := \int_\Theta h(\theta)\,\pi_{n,M,\tau}(d\theta).
\end{equation}
Under Assumptions~\ref{ass:kernel} and~\ref{ass:cont}, the map $\theta\mapsto \Delta(\theta)$ is continuous. Weak continuity of $P_\theta^{\mathrm{obs}}$ and bounded continuity of $k$ imply continuity of all three expectation terms in the expansion Equation~\eqref{eq:mmd-expanded}. Since $\Theta$ is compact, $\Delta$ attains its minimum and $\Theta^\dagger$ is compact.

\begin{theorem}[Monte Carlo consistency]
\label{thm:mc}
Fix $n$, $M$, and $\tau>0$. Let $\theta_1,\dots,\theta_{n_\theta}\stackrel{\mathrm{iid}}{\sim}p(\theta)$, and suppose
$\Delta_{n,M}:\Theta\to[0,\infty)$ is measurable. Let $h:\Theta\to\mathbb R$ be measurable and suppose $\int_\Theta |h(\theta)|\,\pi_{n,M,\tau}(d\theta)<\infty$.
Then $\sum_{j=1}^{n_\theta} w_j h(\theta_j)
\xrightarrow[n_\theta\to\infty]{\mathrm{a.s.}}
\Phi_{n,M,\tau}(h)$.
\end{theorem}

Theorem~\ref{thm:mc} shows that for fixed data, simulation batch size, and temperature, the weighted particle approximation is asymptotically exact as the number of parameter draws increases. We next show that if the empirical discrepancy converges uniformly to the population discrepancy, then the empirical pseudo-posterior approaches the population pseudo-posterior for fixed $\tau$.

\begin{theorem}\label{thm:fixedtau}
Suppose Assumptions~\ref{ass:param}--\ref{ass:uniform} hold, and let $\tau>0$ be fixed. Then for every bounded continuous $h:\Theta\to\mathbb R$, $\int_\Theta h(\theta)\,\pi_{n,M,\tau}(d\theta)
 \xrightarrow[n,M\to\infty]{P} 
\int_\Theta h(\theta)\,\pi_\tau(d\theta)$.
\end{theorem}

The preceding result is the fixed-temperature analogs of posterior consistency \citep{walker2004modern}. The empirical pseudo-posterior inherits the limiting distribution that would have been obtained had the exact population discrepancy $\Delta(\theta)$ been known.

\subsection{Identifiability and Posterior Concentration} 

We now turn to the asymptotic regime in which $n\to\infty$, $M\to\infty$, and the temperature $\tau_{n,M}\downarrow0$. The relevant target is the identified set
\begin{equation}
\Theta^\dagger
=
\arg\min_{\theta\in\Theta}\Delta(\theta).
\end{equation}
Since the exponential weights favor smaller values of $\Delta(\theta)$, one expects $\pi_{n,M,\tau_{n,M}}$ to concentrate around $\Theta^\dagger$ provided the empirical loss approximates $\Delta$ finely enough relative to the shrinking temperature. The next theorem formalizes this.

\begin{theorem}[Posterior Consistency]
\label{thm:consistency}
Suppose Assumptions~\ref{ass:param}--\ref{ass:rate} hold. Let $\tau_{n,M}\downarrow0$. Then for every open neighborhood $U\supset\Theta^\dagger$,
\begin{equation}
\pi_{n,M,\tau_{n,M}}(U)
 \xrightarrow[n,M\to\infty]{P} 1.
\label{eq:post-consistency}
\end{equation}
Equivalently, for every closed set $F\subset\Theta$ satisfying $F\cap\Theta^\dagger=\varnothing$,
\begin{equation}
\pi_{n,M,\tau_{n,M}}(F)
 \xrightarrow[n,M\to\infty]{P} 0.
\label{eq:closed-consistency}
\end{equation}
\end{theorem}

Theorem~\ref{thm:consistency} is the main posterior-consistency result. It states that the empirical pseudo-posterior places asymptotically all of its mass on arbitrary neighborhoods of the minimizers of the population MMD loss. In particular, when the simulator is correctly specified and the kernel is characteristic, the pseudo-posterior concentrates on the set of parameter values whose induced observed-data distributions coincide with the true observed law. This theorem relies on Assumption~\ref{ass:rate}. In Appendix~\ref{sec:rate-justification}, we provide justification for this assumption. 

The previous theorem is stated relative to the minimizer set $\Theta^\dagger$. We now study the consequences under correct specification. Implications under misspecification are outside the scope of this work.

\begin{proposition}
\label{prop:zero}
Under Assumptions~\ref{ass:kernel} and~\ref{ass:correct}, $\inf_{\theta\in\Theta}\Delta(\theta)=0$ and
$\theta^\star\in\Theta^\dagger$.
\end{proposition}

The next result characterizes point identification. The set $\Theta^\dagger$ may nonetheless be non-singleton, reflecting observational non-identifiability after measurement and selection.

\begin{proposition}[Identifiability]
\label{prop:point}
Assume Assumptions~\ref{ass:kernel} and~\ref{ass:correct}. If the map $\theta\longmapsto P_\theta^{\mathrm{obs}}$ is injective on $\Theta$, then $\Theta^\dagger=\{\theta^\star\}$. Consequently, under the assumptions of Theorem~\ref{thm:consistency}, for every open neighborhood $U$ of $\theta^\star$, $\pi_{n,M,\tau_{n,M}}(U)\xrightarrow{P}1$.
\end{proposition}

Proposition~\ref{prop:point} clarifies the sense in which the MMD loss can deliver stronger identification than summary-based methods \citep[e.g.,][]{farahi2026simulation}. If the induced observed-data law uniquely determines the parameter, then the pseudo-posterior is asymptotically point-concentrated. If not, the method still enjoys set consistency and concentrates on the observationally indistinguishable region.

\subsection{Joint Asymptotics with the Monte Carlo Approximation} 

The final result combines the two sources of approximation: finite parameter sampling and finite data or simulation sample sizes.

\begin{corollary}\label{cor:two-stage}
Let $h:\Theta\to\mathbb R$ be bounded and continuous, and let $n_\theta\to\infty$. First, fix $\tau>0$. Under the assumptions of Theorems~\ref{thm:mc} and~\ref{thm:fixedtau},
\begin{equation}
\sum_{j=1}^{n_\theta} w_j h(\theta_j)
-
\int_\Theta h(\theta)\,\pi_{n,M,\tau}(d\theta)
\xrightarrow[n_\theta\to\infty]{\mathrm{a.s.}}0
\label{eq:two-stage-mc}
\end{equation}
for each fixed $(n,M,\tau)$, and
\begin{equation}
\sum_{j=1}^{n_\theta} w_j h(\theta_j)
-
\int_\Theta h(\theta)\,\pi_{\tau}(d\theta)
\xrightarrow[n,M\to\infty]{P}0.
\label{eq:two-stage-fixedtau}
\end{equation}

Second, let $\tau_{n,M}\downarrow 0$. Under the assumptions of Theorems~\ref{thm:mc} and~\ref{thm:consistency}, for every open neighborhood $U\supset\Theta^\dagger$,
\begin{equation}
\Pi_{n_\theta}^{\mathrm{MMD}}(U)\xrightarrow[n,M\to\infty]{P}1.
\label{eq:two-stage-concentration}
\end{equation}
\end{corollary}

The asymptotic behavior follows a two-stage structure. For fixed temperature $\tau>0$, the weighted sample $\Pi_{n_\theta}^{\mathrm{MMD}}$ provides a consistent Monte Carlo approximation to the empirical pseudo-posterior $\pi_{n,M,\tau}$ as $n_\theta\to\infty$, and as $n,M\to\infty$, this empirical pseudo-posterior converges to the population target $\pi_\tau$, which has the form of a Gibbs posterior with loss $\Delta(\theta)$ \citep{bissiri2016general}. In the vanishing-temperature regime $\tau_{n,M}\downarrow0$ and correctly specified setting, the empirical pseudo-posterior concentrates on the true $\theta^\star$, and the same concentration holds for its Monte Carlo approximation, implying that $\Pi_{n_\theta}^{\mathrm{MMD}}$ assigns asymptotically all mass to any neighborhood of $\theta^\star$. 

$\tau$ plays a key role in controlling the concentration and calibration of the resulting pseudo-posterior, and its choice directly impacts empirical performance. In all experiments, we adopt a data-driven heuristic for selecting $\tau$ without tuning (see Appendix~\ref{sec:tau}). This choice adapts to the scale of the discrepancy and the sample size, and was found to yield stable and well-calibrated posterior approximations across a range of settings. While this simple rule performs well in practice (see Appendix~\ref{sec:tau}), the sensitivity of the pseudo-posterior to $\tau$ suggests that more principled or adaptive calibration strategies could further improve performance. Developing a scalable method, for example, based on uncertainty quantification or target coverage properties, remains an open problem for future work.

\section{Benchmarking} \label{sec:benchmarking}

\paragraph{Linear Regression.} We evaluate the proposed method under a controlled simulation setting with heteroscedastic measurement error in both covariates and responses, where latent covariates are drawn from a Gaussian mixture distribution, intrinsic scatter is Gaussian, and measurement error with heteroscedastic Gaussian, Laplace, and Iniform noises. A detailed description of the data-generating process, observation model, and inference procedures is provided in Appendix~\ref{app:EIV_linear}. We compare against a range of standard errors-in-variables baselines, including {\bf LinMix} \cite{kelly2007some}, {\bf Deming regression} \cite{deming1943statistical}, {\bf BCES} \cite{akritas1996linear}, {\bf ODR} \cite{boggs1990orthogonal}, {\bf SIMEX} \cite{cook1994simulation}, and ordinary least squares ({\bf OLS}), covering both hierarchical Bayesian, likelihood-based, and classical correction methods. \vspace{-3mm}

\begin{figure}[hb]
    \centering
    \includegraphics[width=0.90\linewidth]{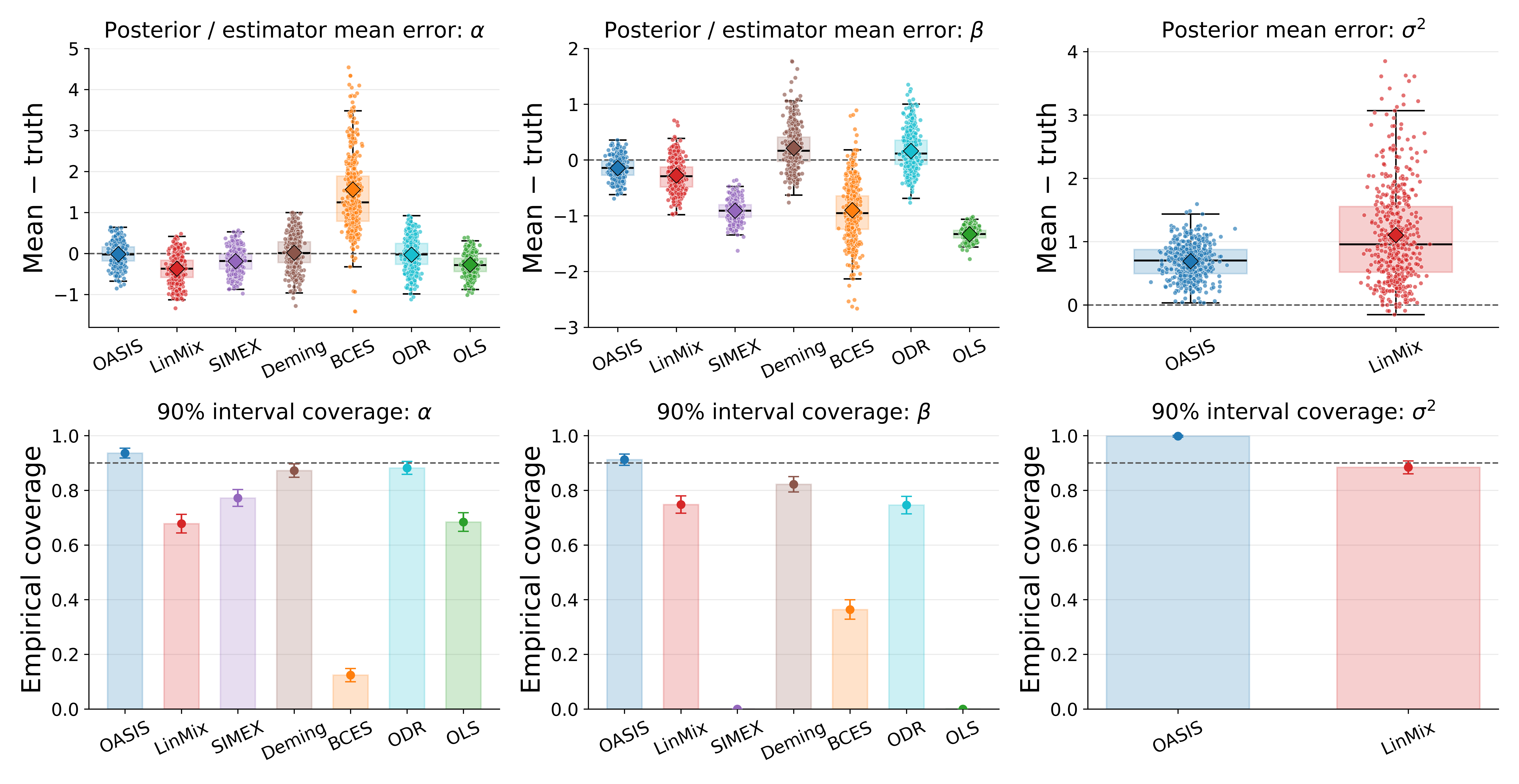}
    \vspace{-3mm}
    \caption{Performance comparison under Gaussian intrinsic noise. \textit{Top row:} distribution of estimation error (posterior or estimator mean minus ground truth) for $\alpha$, $\beta$, and $\sigma^2$ across 500 independent realizations. \textit{Bottom row:} empirical coverage of nominal 90\% uncertainty intervals. The proposed method (\texttt{OASIS}) achieves low bias and near-nominal coverage across parameters, while classical errors-in-variables methods exhibit increased bias and miscalibration under heteroscedastic measurement noise. Full experimental details are provided in Appendix~\ref{app:EIV_linear}.}
    \label{fig:EIV_linear_gaussian}
\end{figure}

Figures~\ref{fig:EIV_linear_gaussian} and ~\ref{fig:EIV_linear_laplace_uniform} summarize performance over repeated realizations in terms of bias and 90\% interval coverage for Gaussian, Laplace, and Uniform measurement errors. The proposed method achieves consistently low bias for both $\alpha$ and $\beta$, while maintaining well-calibrated uncertainty estimates with coverage close to the nominal level. In contrast, classical estimators exhibit varying degrees of bias and miscalibration, particularly under heteroscedastic noise and non-Gaussian covariate structure. Notably, BCES and OLS are highly biased and have poor coverage, while SIMEX and Deming provide partial correction but remain sensitive to noise magnitude and distributional assumptions. LinMix performs competitively when its Gaussian assumptions are satisfied, but shows increased bias in the intrinsic scatter parameter. Overall, these results highlight the robustness of the proposed approach to observational distortions and measurement error distributional assumptions, as it operates directly on observed-data distributions rather than relying on explicit likelihood assumptions. We note that $\tau$ is not tuned in any experiment; we use the heuristic discussed earlier, though with tuning, we expect to see even better results.

\begin{figure}
    \centering
    \includegraphics[width=0.70\linewidth]{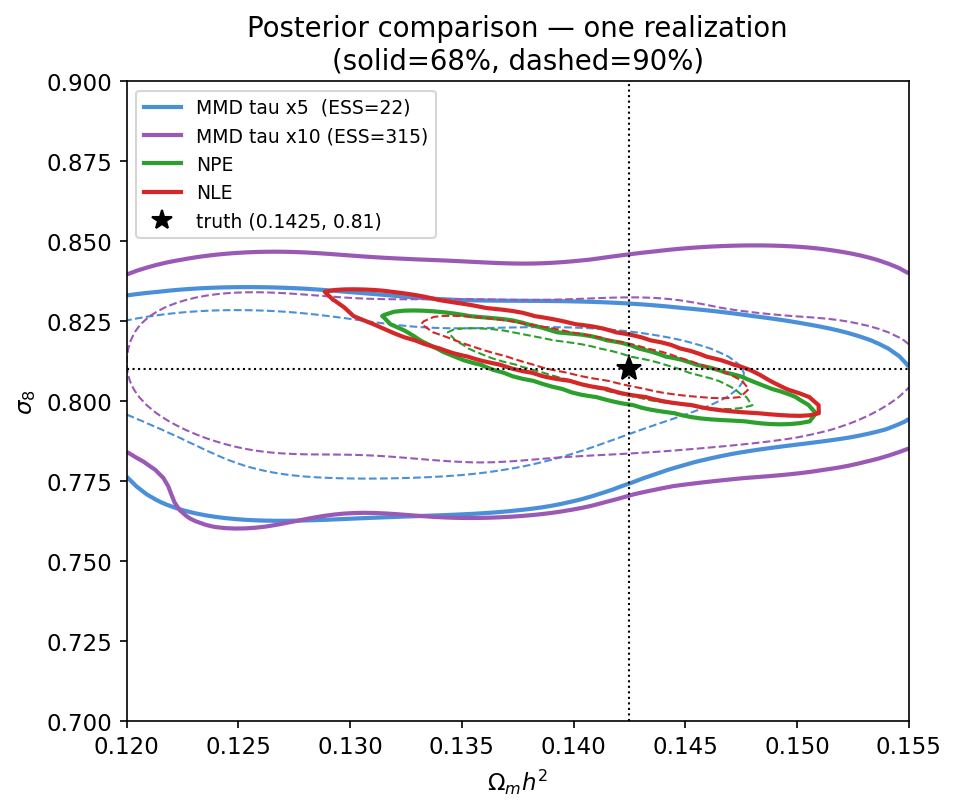}
    \caption{Posterior comparison for a representative realization. Contours show the 68\% and 90\% credible regions for \texttt{OASIS}, NPE, and NLE; the black star marks the true input cosmology.}
    \label{fig:cluster_posterior_main}
\end{figure}

\paragraph{Cluster-cosmology Application.}
We evaluate \texttt{OASIS} on a synthetic multi-wavelength cluster-cosmology benchmark designed to capture the main statistical difficulties of cluster abundance analyses while keeping the inference problem controlled. The simulated cluster survey contains approximately $1{,}700$ clusters over $200\,\mathrm{deg}^2$ and $z\in[0.1,1.0]$, with cosmological parameters restricted to $\theta=(\Omega_m h^2,\sigma_8)$. Halo abundances are generated from the \textsc{MiraTitan} emulator \cite{mcclintock2019aemulus, bocquet2020mira}, and the true cluster population is modeled through multivariate log-linear scaling relations for optical richness, SZe signal-to-noise, weak-lensing mass, and X-ray count rate, including intrinsic scatter and cross-observable correlations \cite{mulroy2019locuss, farahi2019detection}. The observation model includes richness-dependent projection scatter \cite{wu2022optical}, SZe selection \cite{bocquet2024spt}, mass-dependent weak-lensing bias and scatter \cite{grandis2021calibration}, and log-normal X-ray count-rate noise. Selection is modeled through threshold cuts for optical, SZe, and weak-lensing observables, and through a redshift-dependent logistic selection function for X-ray detections. Partial survey overlap is allowed, so missing observables arise from footprint coverage, while non-detections within a footprint are treated as censored data. Full details of the data-generating process are given in Appendix~\ref{app:cluster_dgp}.

\begin{figure}
\centering
    \centering
    \includegraphics[trim={0cm 0cm 0cm 1cm}, clip, width=0.80\linewidth]{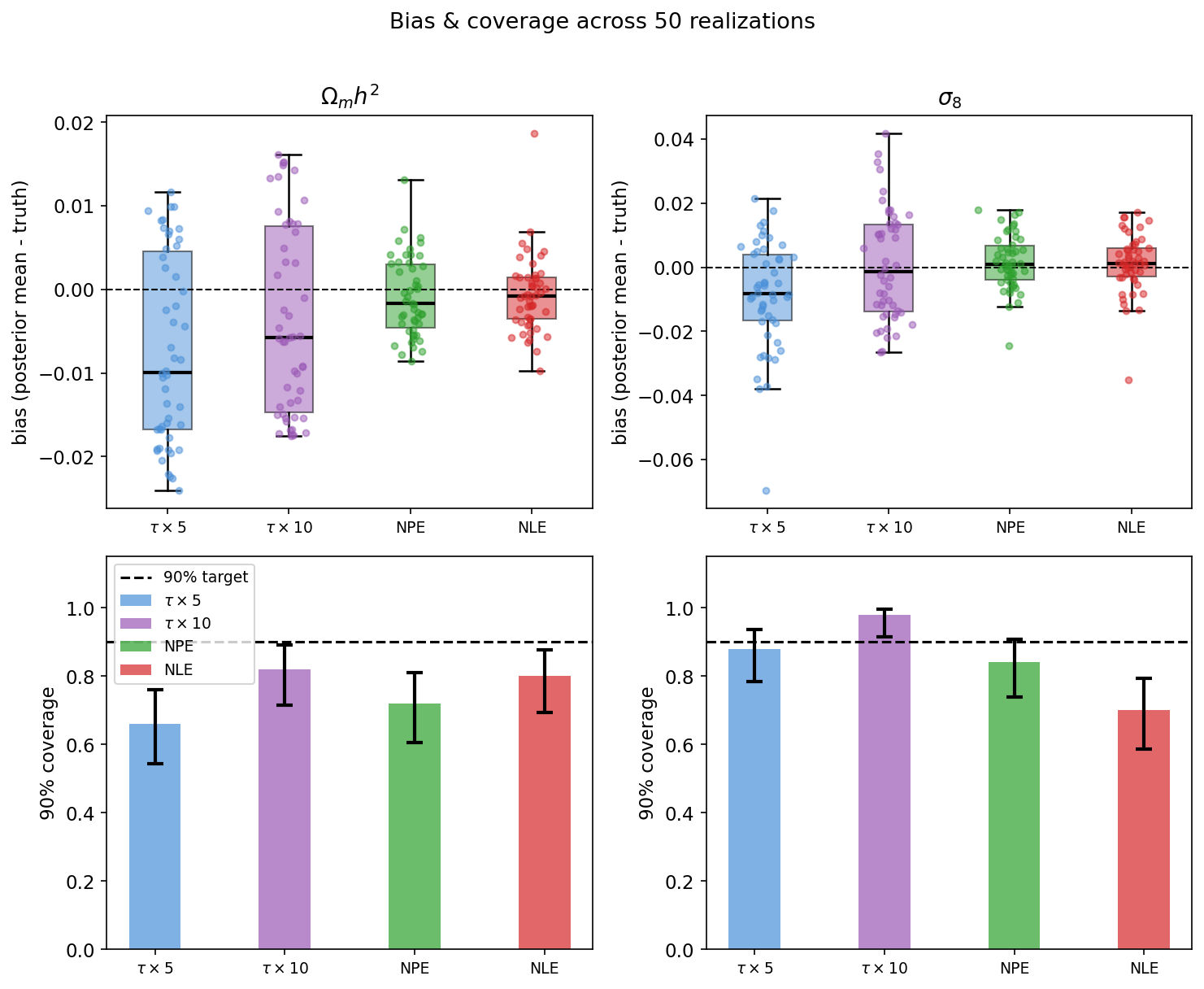}
    \caption{\textit{Top row:} distribution of posterior-mean estimation error across 50 independent realizations for $\Omega_m h^2$ and $\sigma_8$. \textit{Bottom row:} empirical 90\% credible interval coverage with Wilson-score confidence intervals; the dashed line marks nominal 90\% coverage. }
    \label{fig:cosmo_coverage}
\end{figure}

For inference, \texttt{OASIS} draws $2{,}000$ parameter proposals from flat priors $\Omega_m h^2\in[0.12,0.155]$ and $\sigma_8\in[0.70,0.90]$. Each proposal is scored by the squared MMD between observed and simulated five-dimensional observable distributions $(z,\ln\lambda,\xi,\ln\mathcal{M},\ln C_R)$ using an isotropic RBF kernel with median-heuristic bandwidth. We set $\tau=10\hat{\tau}$, where $\hat{\tau}$ is defined in Section~\ref{sec:setup}. We compare against Neural Posterior Estimation (\textbf{NPE}) and Neural Likelihood Estimation (\textbf{NLE}) from the \texttt{sbi} library \citep{tejero2020sbi}, each trained on 5000 simulations using a 50-dimensional hand-designed summary statistic. Performance is evaluated over 50 independent realizations with true cosmologies drawn from the prior.

Figure~\ref{fig:cluster_posterior_main} shows that all methods primarily constrain $\sigma_8$, while $\Omega_m h^2$ remains weakly identified, as expected for cluster abundance in this mass and redshift range. \texttt{OASIS} contains the true parameter values within its 68\% credible contour and agrees broadly with the neural methods along the constrained $\sigma_8$ direction, while maintaining a more conservative contour in $\Omega_m h^2$. Across 50 realizations, Figure~\ref{fig:cosmo_coverage} shows that \texttt{OASIS} has negligible mean bias, with $\Delta\Omega_m h^2=-0.004$ and $\Delta\sigma_8=0.001$, and achieves empirical 90\% coverage of 82\% for $\Omega_m h^2$ and 98\% for $\sigma_8$. NPE and NLE have similarly small bias but lower coverage, indicating overconfident posterior uncertainty rather than systematic displacement. These results suggest that direct distributional comparison with MMD can potentially provide better-calibrated uncertainty quantification in this benchmark without requiring summary-statistic engineering or amortized neural training.

We emphasize that this comparison is not intended to optimize every source of cosmological information available in the simulated survey. NPE and NLE can benefit from hand-crafted summary statistics that explicitly encode both the shape and amplitude of the halo mass function, including cluster-count features \cite{ho2024ltu}. By contrast, the present \texttt{OASIS} implementation compares normalized observable distributions and therefore uses primarily shape information in the induced cluster population, rather than the full abundance normalization. Some information loss is therefore expected, particularly for $\Omega_m h^2$, whose posterior intervals may be inflated relative to methods that directly use abundance information (see Appendix~\ref{app:distribution_vs_abundance} for full argument). The purpose of this experiment is instead to show that \texttt{OASIS} can handle the full complexity of a realistic multi-wavelength cluster-cosmology analysis, including observational noise, correlated observables, and survey selection effects, while still delivering competitive and well-calibrated constraints even under a deliberately suboptimal information representation.

\section{Conclusion} \label{sec:conclusion}

\texttt{OASIS} provides an observation-aware framework for simulation-based inference in scientific problems where the relevant comparison is between observed data and forward-simulated observations, rather than between data and idealized latent simulator outputs. By embedding measurement error, selection effects, and survey-specific transformations directly into the simulation loop, and by replacing handcrafted summaries with an MMD-based distributional discrepancy, \texttt{OASIS} offers a flexible pseudo-posterior construction with theoretical consistency guarantees and empirically well-calibrated uncertainty. The experiments show that this approach is robust in controlled errors-in-variables settings and remains competitive in a realistic multi-wavelength cluster-cosmology benchmark involving correlated observables, incomplete coverage, and selection effects. More broadly, the results suggest that distributional matching with explicit observational modeling can serve as a practical alternative to summary-based or neural SBI methods when likelihoods are unavailable, observational pipelines are complex, and calibrated uncertainty is more important than aggressive compression or amortized inference. Future work should focus on principled temperature calibration, scalable kernel approximations, and extensions that incorporate count or abundance information directly into the discrepancy, thereby improving efficiency while preserving the observation-aware structure of the method.

\paragraph{Limitations.} The proposed framework relies on the availability of a known or accurately quantifiable observation model, which may be restrictive in settings where measurement processes are poorly understood or only partially specified. While the method avoids explicit likelihood evaluation, it incurs nontrivial computational cost due to repeated forward simulations and pairwise kernel evaluations, which can scale quadratically with sample size and simulation budget. The theoretical guarantees rely on correct specification and identifiability of the observed-data mapping; under model misspecification or observational non-identifiability, the method may concentrate on a set of observationally equivalent parameters rather than a unique ground truth. Performance is also sensitive to the choice of kernel and its hyperparameters, which govern the geometry of distributional comparison and may affect identifiability in practice. Although MMD provides a powerful distributional metric, it may have reduced sensitivity to certain high-frequency discrepancies or structured differences in high-dimensional settings without careful kernel design. The generalization to multinomial data requires careful consideration and is an interesting avenue for further investigation \cite{koslovsky2026unified}.

\section*{Acknowledgment}
This work was supported by the NSF under Cooperative Agreement 2421782, and the Simons Foundation grant MPS-AI-00010515 awarded to the NSF-Simons AI Institute for Cosmic Origins (CosmicAI, \url{https://www.cosmicai.org/}).

\bibliographystyle{unsrt}
\bibliography{bibliography, zhou_bib}


\appendix

\section{Assumptions} \label{app:assumptions}

Here, we state all the regularity conditions used throughout this work.

\begin{assumption}[Parameter space and prior]
\label{ass:param}
The parameter space $\Theta\subset\mathbb R^p$ is compact. The prior density $p(\theta)$ is continuous on $\Theta$ and strictly positive on $\Theta$.
\end{assumption}

Assumption~\ref{ass:param} guarantees existence of minimizers of the population discrepancy and rules out trivial failures of posterior concentration caused by a prior that excludes the target region.

\begin{assumption}[Kernel regularity]
\label{ass:kernel}
The kernel $k:\mathbb R^d\times\mathbb R^d\to\mathbb R$ is bounded, continuous, and positive definite. Write
\begin{equation}
\|k\|_\infty := \sup_{y,y'\in\mathbb R^d}|k(y,y')|<\infty.
\end{equation}
Moreover, $k$ is characteristic.
\end{assumption}

Boundedness guarantees that MMD is well-defined for all probability measures on $\mathbb R^d$, while characteristicness ensures that the population discrepancy vanishes only when the corresponding distributions coincide.

\begin{assumption}
\label{ass:cont}
The map $\theta\mapsto P_\theta^{\mathrm{obs}}$ is weakly continuous on $\Theta$. Equivalently, for every bounded continuous function $f:\mathbb R^d\to\mathbb R$, the map
\begin{equation}
\theta\longmapsto \int f(y)\,P_\theta^{\mathrm{obs}}(dy)
\end{equation}
is continuous.
\end{assumption}

This assumption is natural for simulators whose latent generative law and observation operator vary continuously with the physical parameters. It implies continuity of the kernel mean embedding and hence continuity of the MMD loss itself.

\begin{assumption}
\label{ass:uniform}
As $n,M\to\infty$,
\begin{equation}
\sup_{\theta\in\Theta}
\bigl|\Delta_{n,M}(\theta)-\Delta(\theta)\bigr|
\;\xrightarrow{P}\;0.
\label{eq:uniform-consistency-ass}
\end{equation}
\end{assumption}

Assumption~\ref{ass:uniform} is the key stochastic approximation condition. It can be verified under standard empirical-process conditions for bounded kernels and a sufficiently regular simulator family. In the present exposition we state it explicitly, since it isolates the statistical content needed for the pseudo-posterior arguments from the details of the simulator implementation.

\begin{assumption}
\label{ass:ident}
Define the MMD-identified set
\begin{equation}
\Theta^\dagger
:=
\arg\min_{\theta\in\Theta}\Delta(\theta).
\end{equation}
This set is nonempty, and for every open neighborhood $U\supset \Theta^\dagger$ there exists a separation constant $\eta(U)>0$ such that
\begin{equation}
\inf_{\theta\notin U}\Delta(\theta)
\ge
\inf_{\vartheta\in\Theta}\Delta(\vartheta)+\eta(U).
\label{eq:gap}
\end{equation}
\end{assumption}

Assumption~\ref{ass:ident} is a standard gap condition. Since $\Theta$ is compact and $\Delta$ is continuous, the identified set is nonempty; the additional statement simply requires strict separation of points outside any neighborhood of the minimizer set.

\begin{assumption}
\label{ass:rate}
Let $\tau=\tau_{n,M}>0$ satisfy $\tau_{n,M}\downarrow0$ and
\begin{equation}
\sup_{\theta\in\Theta} \bigl|\Delta_{n,M}(\theta)-\Delta(\theta)\bigr| = o_P(\tau_{n,M}^2).
\label{eq:rate-cond}
\end{equation}
\end{assumption}

Assumption~\ref{ass:rate} strengthens Assumption~\ref{ass:uniform} in the regime where the temperature vanishes. It implies that the stochastic error in the empirical loss must be asymptotically negligible relative to the scale $\tau_{n,M}^2$ appearing in the exponential weights. The justification for this assumption is provided in Appendix~\ref{sec:rate-justification}.

\begin{assumption}[Correct specification]
\label{ass:correct}
There exists $\theta^\star\in\Theta$ such that
\begin{equation}
P_{\theta^\star}^{\mathrm{obs}}=P_{\mathrm{obs}}.
\end{equation}
\end{assumption}

\section{Justification of the Rate Assumption for Empirical MMD}
\label{sec:rate-justification}

Here, we provide a justification for Assumption~\ref{ass:rate}, namely that the empirical MMD discrepancy converges uniformly to its population counterpart at a rate sufficient to ensure
\begin{equation}
\sup_{\theta\in\Theta}
\bigl|\Delta_{n,M}(\theta)-\Delta(\theta)\bigr|
=
o_P(\tau_{n,M}^2).
\end{equation}
We show that this condition follows from standard concentration results for bounded kernels together with mild regularity conditions on the simulator family, provided the temperature sequence $\tau_{n,M}$ is chosen appropriately.

\subsection{Setup and decomposition}

Recall that for each $\theta\in\Theta$,
\begin{equation}
\Delta(\theta)
=
\mathrm{MMD}_k^2\!\bigl(P_{\mathrm{obs}},P_\theta^{\mathrm{obs}}\bigr),
\qquad
\Delta_{n,M}(\theta)
=
\mathrm{MMD}_k^2\!\bigl(\widehat P_{\mathrm{obs}},\widehat P_{\theta,M}^{\mathrm{obs}}\bigr),
\end{equation}
where $\widehat P_{\mathrm{obs}}$ is based on $n$ observed samples and $\widehat P_{\theta,M}^{\mathrm{obs}}$ is based on $M$ simulated samples.

Using the expansion of squared MMD, we may write
\begin{align}
\Delta(\theta)
&=
\mathbb E[k(Y,Y')]
+
\mathbb E[k(Z_\theta,Z_\theta')]
-
2\mathbb E[k(Y,Z_\theta)],
\\
\Delta_{n,M}(\theta)
&=
\frac{1}{n(n-1)}\sum_{i\neq i'} k(Y_i,Y_{i'})
+
\frac{1}{M(M-1)}\sum_{m\neq m'} k(Z_{\theta,m},Z_{\theta,m'})
-
\frac{2}{nM}\sum_{i,m} k(Y_i,Z_{\theta,m}).
\end{align}

Thus,
\begin{equation}
\Delta_{n,M}(\theta)-\Delta(\theta)
=
A_n + B_M(\theta) + C_{n,M}(\theta),
\label{eq:decomposition}
\end{equation}
where:
\begin{align}
A_n &= \text{error in the observed-observed term}, \\
B_M(\theta) &= \text{error in the simulated-simulated term}, \\
C_{n,M}(\theta) &= \text{error in the cross term}.
\end{align}

We now control each term uniformly in $\theta$.

\begin{assumption}[Uniform Lipschitz regularity]
\label{ass:lipschitz}
There exists $L>0$ such that for all $\theta,\theta'\in\Theta$,
\begin{equation}
\sup_{y\in\mathbb R^d}
\left|
\mathbb E_{Z\sim P_\theta^{\mathrm{obs}}} k(y,Z)
-
\mathbb E_{Z\sim P_{\theta'}^{\mathrm{obs}}} k(y,Z)
\right|
\le L\|\theta-\theta'\|.
\end{equation}
\end{assumption}

This assumption ensures that the class of functions indexed by $\theta$ has controlled complexity.

We begin with pointwise bounds.

\begin{lemma}\label{lem:pointwise}
For each fixed $\theta$,
\begin{equation}
|\Delta_{n,M}(\theta)-\Delta(\theta)| = O_P\!\left(n^{-1/2} + M^{-1/2}\right).
\end{equation}
\end{lemma}

\begin{proof}
Each term in \eqref{eq:decomposition} is an average of bounded random variables.

For the observed-observed term $A_n$, standard U-statistic concentration for bounded kernels yields
\begin{equation}
|A_n| = O_P(n^{-1/2}).
\end{equation}

Similarly, for the simulated term $B_M(\theta)$,
\begin{equation}
|B_M(\theta)| = O_P(M^{-1/2}).
\end{equation}

For the cross term $C_{n,M}(\theta)$, which is an average of $nM$ bounded terms,
\begin{equation}
|C_{n,M}(\theta)| = O_P(n^{-1/2} + M^{-1/2}).
\end{equation}

Combining these bounds yields the result.
\end{proof}

For a similar result, see Theorem~10 in \citep{gretton2012kernel}.

We now extend the pointwise bound to a uniform bound over $\Theta$.

\begin{lemma}
\label{lem:uniform-rate}
Under Assumptions~\ref{ass:param}, \ref{ass:kernel}, \ref{ass:cont} and \ref{ass:lipschitz},
\begin{equation}
\sup_{\theta\in\Theta}
|\Delta_{n,M}(\theta)-\Delta(\theta)|
=
O_P\!\left(n^{-1/2} + M^{-1/2}\right).
\end{equation}
\end{lemma}

\begin{proof}
Since $\Theta$ is compact, for any $\varepsilon>0$ there exists a finite $\varepsilon$-net $\{\theta_1,\dots,\theta_N\}$ covering $\Theta$, with $N \lesssim \varepsilon^{-p}$.

By the pointwise bound (Lemma~\ref{lem:pointwise}),
\begin{equation}
\max_{1\le j\le N}
|\Delta_{n,M}(\theta_j)-\Delta(\theta_j)|
=
O_P(n^{-1/2}+M^{-1/2}).
\end{equation}

Now for any $\theta$, choose $\theta_j$ such that $\|\theta-\theta_j\|\le\varepsilon$. Using the Lipschitz assumption and boundedness of the kernel expectations, one can show
\begin{equation}
|\Delta_{n,M}(\theta)-\Delta_{n,M}(\theta_j)| \le C\varepsilon\quad{\rm and}\quad
|\Delta(\theta)-\Delta(\theta_j)|
\le C\varepsilon
\end{equation}
for some constant $C$.

Therefore,
\begin{equation}
|\Delta_{n,M}(\theta)-\Delta(\theta)|
\le
|\Delta_{n,M}(\theta_j)-\Delta(\theta_j)| + 2C\varepsilon.
\end{equation}

Taking the supremum over $\theta\in\Theta$ and then letting $\varepsilon\to 0$ yields
\begin{equation}
\sup_{\theta\in\Theta}
|\Delta_{n,M}(\theta)-\Delta(\theta)|
=
O_P(n^{-1/2}+M^{-1/2}).
\end{equation}
\end{proof}

We now connect the uniform rate to Assumption~\ref{ass:rate}.

\begin{theorem}
\label{thm:rate-verify}
Suppose Assumptions~\ref{ass:param}, \ref{ass:kernel}, \ref{ass:cont} and \ref{ass:lipschitz} hold. If the temperature sequence $\tau_{n,M}$ satisfies
\begin{equation}
n^{-1/2} + M^{-1/2} = o(\tau_{n,M}^2),
\label{eq:tau-condition}
\end{equation}
then
\begin{equation}
\sup_{\theta\in\Theta}
|\Delta_{n,M}(\theta)-\Delta(\theta)|
=
o_P(\tau_{n,M}^2).
\end{equation}
\end{theorem}

\begin{proof}
By Lemma~\ref{lem:uniform-rate},
\begin{equation}
\sup_{\theta\in\Theta}
|\Delta_{n,M}(\theta)-\Delta(\theta)|
=
O_P(n^{-1/2}+M^{-1/2}).
\end{equation}
Let $a_{n,M}=n^{-1/2}+M^{-1/2}$. Then
\begin{equation}
\frac{\sup_\theta |\Delta_{n,M}(\theta)-\Delta(\theta)|}{\tau_{n,M}^2}
=
O_P\!\left(\frac{a_{n,M}}{\tau_{n,M}^2}\right).
\end{equation}
Under condition \eqref{eq:tau-condition}, $a_{n,M}/\tau_{n,M}^2 \to 0$ \citep[see Section 2.2 of][]{vd1998asymptotic}, hence
\begin{equation}
\frac{\sup_\theta |\Delta_{n,M}(\theta)-\Delta(\theta)|}{\tau_{n,M}^2}
\xrightarrow{P}0,
\end{equation}
which proves the result. 
\end{proof}

Theorem~\ref{thm:rate-verify} shows that Assumption~\ref{ass:rate} is not a primitive requirement but follows from two components. First, a uniform stochastic rate for empirical MMD, which is a consequence of bounded kernels and regularity of the simulator family, and second, a calibration condition on the temperature sequence $\tau_{n,M}$.

We note that no deterministic boundedness of $\Delta_{n,M}$ is required. The empirical MMD estimator may exhibit large deviations with small probability, but as long as these deviations are controlled at the $O_P(n^{-1/2}+M^{-1/2})$ scale, they are asymptotically negligible relative to $\tau_{n,M}^2$ under condition \eqref{eq:tau-condition}. This is sufficient for the posterior concentration results established in Section~\ref{sec:theory}.

\section{\texttt{OASIS} Algorithm}

Algorithm~\ref{alg:OASIS} implements a likelihood-free inference procedure based on distributional comparison between observed data and forward simulations. The key feature is the explicit two-stage simulation: latent samples are first drawn from the mechanistic simulator and are subsequently propagated through the observation model. This produces synthetic datasets that are directly comparable to the observed data at the level of the observed-data law.

The discrepancy between observed and simulated datasets is quantified using MMD, which operates on empirical measures and avoids the need for handcrafted summary statistics. The resulting loss $\Delta_{n,M}(\theta)$ is stochastic due to the finite simulation budget $M$, inducing a pseudo-posterior indexed by both $M$ and the temperature parameter $\tau$. Larger values of $M$ reduce Monte Carlo variability in the loss, while $\tau$ controls the concentration of the posterior around minimizers of the population discrepancy.

The algorithm can be interpreted as a particle-based approximation to a Gibbs-type posterior defined through a distributional loss. It is straightforward to parallelize across particles and simulation draws, and can be combined with adaptive proposals or sequential schemes to improve efficiency. In practice, the computational cost is dominated by simulator evaluations and scales as $\mathcal O(n_\theta M)$.

\begin{algorithm}[t]
\caption{\texttt{OASIS}: MMD-based Simulation-Based Inference with Observational Modeling}
\label{alg:OASIS}
\begin{algorithmic}[1]

\Require Observed data $\mathcal D_{\mathrm{obs}}=\{y_i^{\mathrm{obs}}\}_{i=1}^n$, prior $p(\theta)$, simulator $P_\theta^{\mathrm{true}}$, observation model $P^{\rm Err}_\theta$, kernel $k$, number of particles $n_\theta$, simulation budget $M$, temperature $\tau$
\Ensure Weighted sample $\{(\theta_j,w_j)\}_{j=1}^{n_\theta}$ approximating $\pi_\tau(\theta)$

\State Compute empirical measure $\widehat P_{\mathrm{obs}} = \frac{1}{n}\sum_{i=1}^n \delta_{y_i^{\mathrm{obs}}}$

\For{$j = 1$ to $n_\theta$}
    \State Draw $\theta_j \sim p(\theta)$
    
    \For{$m = 1$ to $M$}
        \State Draw $u_{j,m}^{\mathrm{true}} \sim P_{\theta_j}^{\mathrm{true}}$
        \State Draw $y_{j,m}^{\mathrm{sim,obs}} \sim P^{\rm Err}_{\theta_j}(\cdot \mid u_{j,m}^{\mathrm{true}})$
    \EndFor
    
    \State Compute discrepancy
    \[
    \Delta_{n,M}(\theta_j) = \mathrm{MMD}_k^2\!\left(\widehat P_{\mathrm{obs}}, \widehat P_{\theta_j,M}^{\mathrm{sim,obs}}\right)
    \]
    
    \State Set unnormalized weight
    \[
    \widetilde w_j = \exp\!\left(-\frac{\Delta_{n,M}(\theta_j)}{2\tau^2}\right)
    \]
\EndFor

\State Normalize weights
\[
w_j = \frac{\widetilde w_j}{\sum_{k=1}^{n_\theta} \widetilde w_k}
\]

\State \Return $\{(\theta_j, w_j)\}_{j=1}^{n_\theta}$

\end{algorithmic}
\end{algorithm}

\section{Temperature Selection}
\label{sec:tau}

The temperature parameter $\tau$ controls the concentration of the pseudo-posterior
\begin{equation*}
\pi_\tau(d\theta) \propto p(\theta)\exp\{-\Delta(\theta)/(2\tau^2)\}d\theta,    
\end{equation*}
and therefore governs the trade-off between fidelity to the data and robustness to Monte Carlo variability. Small values of $\tau$ concentrate mass near minimizers of the discrepancy $\Delta(\theta)$, while larger values yield a more diffuse posterior that accounts for uncertainty induced by finite sample size and simulation noise.

\paragraph{Heuristic calibration.}
In practice, we select $\tau$ adaptively based on the empirical scale of the discrepancy across parameter draws. Given a set of particles $\{\theta_j\}_{j=1}^{n_\theta}$ and their associated empirical losses $\{\Delta_{n,M}(\theta_j)\}$, and assuming $n \approx M$ in our experiments, we set
\begin{equation}
\tau = \frac{\sqrt{2\,\mathrm{median}\bigl(\{\Delta_{n,M}(\theta_j)\}\bigr) + \epsilon}}{\sqrt{n}},
\label{eq:tau}
\end{equation}
where $n$ is the size of the observed dataset and $\epsilon>0$ is a small constant for numerical stability.

This choice can be interpreted as calibrating $\tau$ to the typical scale of the discrepancy under the empirical distribution of parameter draws. The use of the median provides robustness to outliers and heavy-tailed behavior in $\Delta_{n,M}(\theta)$, which can arise from poorly fitting parameter values. The factor of $2$ matches the scaling in the pseudo-likelihood $\exp\{-\Delta/(2\tau^2)\}$, yielding weights of order $\exp(-1)$ for parameters whose discrepancy is close to the median.

The normalization by $\sqrt{n}$ reflects the fact that $\Delta_{n,M}(\theta)$ is computed from empirical measures based on $n$ observations. Under standard empirical process considerations, discrepancies such as MMD scale as $O_p(n^{-1/2})$ around their population counterparts. Dividing by $\sqrt{n}$ therefore aligns $\tau$ with the sampling variability of the empirical discrepancy, preventing under-concentration and over-concentration of the pseudo-posterior as the dataset size increases.

From a statistical perspective, $\tau$ plays a role analogous to a noise scale or uncertainty parameter in likelihood-based inference. The proposed heuristic calibrates $\tau$ to reflect both (i) the intrinsic variability of the observed data through the $n^{-1/2}$ scaling, and (ii) the variability of the discrepancy across parameter values through the empirical distribution of $\Delta_{n,M}(\theta)$. As a result, parameters whose induced observed-data distributions are within the typical fluctuation scale of the empirical distribution receive non-negligible weight, while parameters with substantially larger discrepancies are exponentially downweighted.

In our experiments, this heuristic produced stable posterior approximations across a range of simulation budgets and data regimes without additional tuning. We found the median-based rule to provide a simple, robust default that adapts automatically to the problem's scale. Future work can study alternative strategies, such as cross-validation or ESS-based calibration in SMC, to optimize for $\tau$. While our posteriors perform competitively compared to established methods, further performance improvements are expected under a more rigorous tuning strategy. 

\paragraph{Sensitivity to the temperature parameter \texorpdfstring{$\tau$}{tau}.} The temperature parameter $\tau$ governs the concentration of the pseudo-posterior by controlling how strongly the discrepancy $\Delta(\theta)$ influences the weights
\begin{equation}
w(\theta) \;\propto\; \exp\!\left(-\frac{\Delta(\theta)}{2\tau^2}\right).
\end{equation}
Small values of $\tau$ concentrate mass on a small set of low-discrepancy simulations, whereas large values flatten the weighting scheme and yield a more diffuse posterior. In this sense, $\tau$ directly controls a bias–variance tradeoff in the resulting inference procedure.

This role closely mirrors temperature scaling in generalized Bayesian inference, where $\tau$ effectively rescales the loss function and, in turn, the implied likelihood. From this perspective, $\tau$ determines how aggressively the method treats the discrepancy as a surrogate likelihood, with smaller values corresponding to a sharper, more data-dominated posterior and larger values yielding a more conservative, prior-influenced distribution.

In all experiments, we use the data-driven heuristic
\[
\tau_{\mathrm{heuristic}}
\;=\;
\frac{\sqrt{2\,\mathrm{median}\{\Delta(\theta_j)\}}}{\sqrt{n}},
\]
which adapts to both the scale of the discrepancy and the sample size. This normalization allows the exponent in the pseudo-posterior to remain well-scaled as $n$ grows and as the magnitude of $\Delta(\theta)$ varies across problems, thereby avoiding degenerate regimes without requiring manual tuning.

To study sensitivity to $\tau$, we use the EIV experiment, a linear regression with measurement error in both the covariate and response variables, as discussed in Appendix~\ref{app:EIV_linear}. Figure~\ref{fig:tau_sensitivity} shows the effect of varying $\tau$ over several orders of magnitude relative to this heuristic value. The top row reports posterior means for the intercept, slope, and intrinsic variance, while the bottom row shows empirical $90\%$ coverage across repeated realizations.

Several consistent patterns emerge. For very small $\tau$ (i.e., $\tau \ll \tau_{\mathrm{heuristic}}$), the pseudo-posterior becomes overly concentrated. In this regime, the weighting scheme effectively selects a small number of simulations with minimal discrepancy, making the estimator highly sensitive to Monte Carlo variability. This leads to unstable point estimates and severe undercoverage, with empirical coverage near zero across all parameters. The behavior is analogous to a hard minimum-discrepancy estimator and reflects overfitting to simulation noise.

As $\tau$ increases toward the heuristic scale, both estimation and calibration improve sharply. Around $\tau \approx \tau_{\mathrm{heuristic}}$, the posterior means stabilize and align closely with the ground truth, while empirical coverage rises rapidly toward the nominal $90\%$ level. This transition is notably steep, indicating that the heuristic choice lies near a critical regime where the pseudo-posterior achieves a balance between concentration and dispersion.

For large $\tau$ (i.e., $\tau \gg \tau_{\mathrm{heuristic}}$), the pseudo-posterior becomes increasingly diffuse. In this regime, the weights vary only weakly across parameter draws, so inference is effectively governed by the prior sampling distribution rather than the discrepancy. Consequently, coverage approaches or exceeds the nominal level (often becoming conservative), while point estimates may exhibit mild bias due to insufficient concentration around the true parameters. This effect is particularly visible for the intrinsic variance, which remains systematically overestimated in the high-temperature regime.

Figure~\ref{fig:tau_sensitivity} demonstrates that $\tau$ is a genuine calibration parameter that influences both estimation accuracy and uncertainty quantification. The proposed heuristic performs well in practice by selecting a value near the transition between underconcentrated and overconcentrated regimes, yielding stable estimates and near-nominal coverage across parameters. At the same time, the sharp dependence of coverage on $\tau$ suggests that more principled calibration strategies (e.g., choosing $\tau$ to target nominal frequentist coverage, minimize predictive risk, or satisfy information theoretic criteria) could further improve performance. Developing such adaptive procedures in a computationally efficient manner remains an important direction for future work.

\begin{figure}[t]
    \centering
    \includegraphics[width=\textwidth]{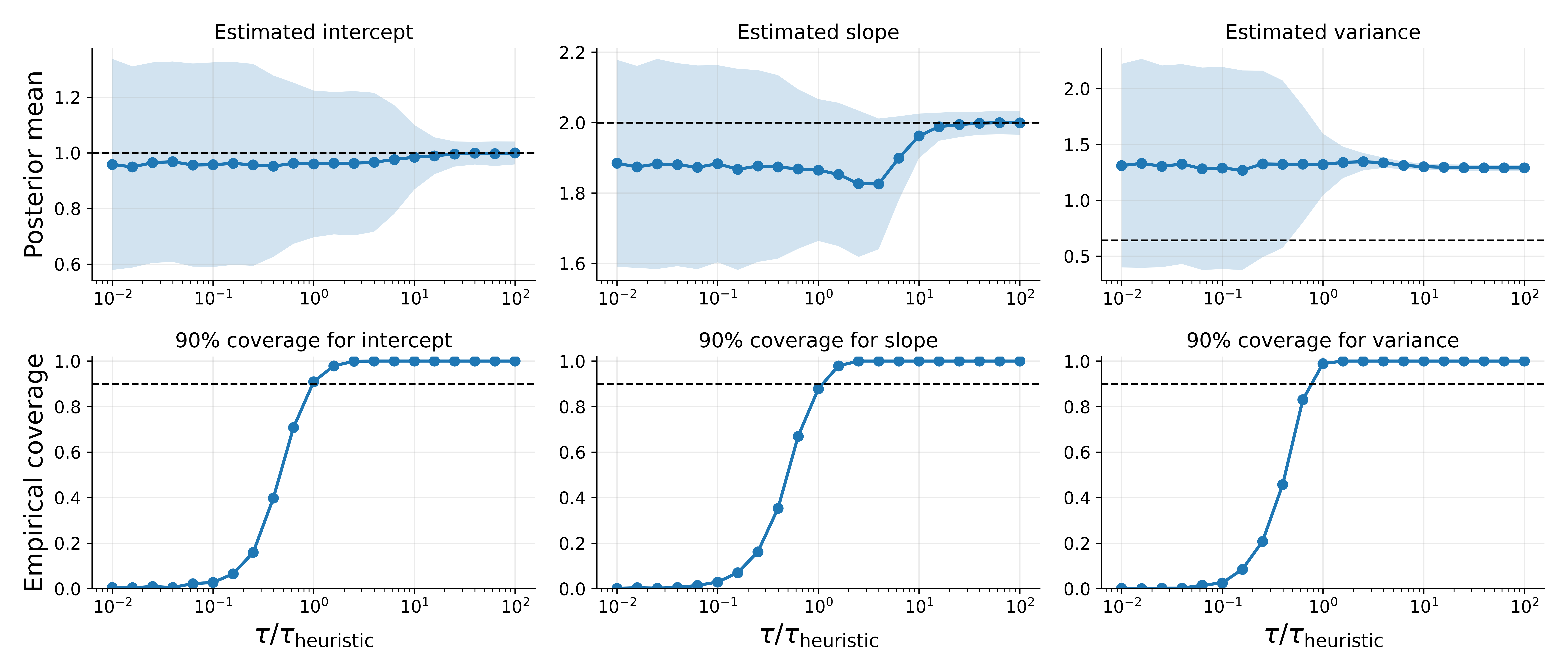}
    \caption{Sensitivity of the OASIS pseudo-posterior to the temperature parameter $\tau$, shown relative to the heuristic choice $\tau_{\mathrm{heuristic}}$. Top row: posterior means for the intercept, slope, and intrinsic variance, with dashed lines indicating ground truth. Bottom row: empirical $90\%$ coverage across repeated realizations. Small $\tau$ leads to overconcentration and severe undercoverage, while large $\tau$ yields diffuse posteriors with conservative coverage. The heuristic choice lies near the transition region where estimates stabilize, and coverage approaches the nominal level.}
    \label{fig:tau_sensitivity}
\end{figure}

\section{Problem Setup}
\label{sec:setup_app}

We consider an observed dataset
\begin{equation}
\mathcal D_{\mathrm{obs}}=\{y_i^{\mathrm{obs}}\}_{i=1}^n,
\qquad
y_i^{\mathrm{obs}}\in\mathbb R^d,
\end{equation}
drawn from an observed-data distribution $P_{\mathrm{obs}}$ on $\mathbb R^d$. Throughout, the vector $y_i^{\mathrm{obs}}$ is understood to contain all quantities retained after the observational pipeline has acted on the underlying population. In particular, we do not distinguish between predictors and responses, nor do we impose an auxiliary regression structure. The inferential goal is instead to compare the full empirical distribution of the observed sample with the distribution induced by a mechanistic simulator under latent parameters $\theta$.

This formulation is particularly natural in scientific settings where the data arise through a two-stage process. First, a latent or ``true'' population is generated according to a structural model indexed by $\theta$. Second, observational noise, survey incompleteness, thresholding, censoring, or other selection effects distort that latent population before the final catalog is recorded. In such problems, the scientifically meaningful comparison is not between the observed sample and the latent simulator output directly, but between the observed sample and the simulator output after the same observational effects have been applied. Our construction is designed precisely for this setting.

To formalize this idea, let $\theta\in\Theta\subseteq\mathbb R^p$ denote the simulator parameter, endowed with prior distribution $p(\theta)$. For each $\theta$, the simulator proceeds in two phases. In the first phase, it generates latent population-level variables
\begin{equation}
u^{\mathrm{true}} \mid \theta \sim P_{\theta}^{\mathrm{true}},
\end{equation}
where $u^{\mathrm{true}}$ denotes the idealized noiseless state. In the second phase, the observational mechanism acts on $u^{\mathrm{true}}$ to produce an observed realization
\begin{equation}
y^{\mathrm{sim,obs}} \mid u^{\mathrm{true}},\theta \sim Q_\theta(\,\cdot\,\mid u^{\mathrm{true}}).
\end{equation}
Here $Q_\theta$ is the observation operator, which may encode measurement noise, heteroscedastic scatter, missingness, truncation, catalog cuts, detection thresholds, or any other survey-selection mechanism. Integrating out the latent population yields the \emph{induced observed-data law}
\begin{equation}
P_\theta^{\mathrm{obs}}(A)
=
\int Q_\theta(A\mid u)\,P_\theta^{\mathrm{true}}(du),
\qquad
A\subseteq \mathbb R^d.
\end{equation}
The statistical task is then to identify those values of $\theta$ for which the induced observed law $P_\theta^{\mathrm{obs}}$ is close to the empirical law underlying $\mathcal D_{\mathrm{obs}}$.

The empirical measure of the observed sample is
\begin{equation}
\widehat P_{\mathrm{obs}}
=
\frac{1}{n}\sum_{i=1}^n \delta_{y_i^{\mathrm{obs}}},
\label{eq:empirical-obs}
\end{equation}
where $\delta_y$ denotes a Dirac measure at $y$. For each proposed parameter value $\theta$, we generate a simulated observed batch of size $M$,
\begin{equation}
\mathcal D_{\theta,M}^{\mathrm{sim,obs}}
=
\{y_{\theta,1}^{\mathrm{sim,obs}},\dots,y_{\theta,M}^{\mathrm{sim,obs}}\},
\qquad
y_{\theta,m}^{\mathrm{sim,obs}}\stackrel{\mathrm{iid}}{\sim}P_\theta^{\mathrm{obs}},
\end{equation}
and define its empirical measure by
\begin{equation}
\widehat P_{\theta,M}^{\mathrm{sim,obs}}
=
\frac{1}{M}\sum_{m=1}^M \delta_{y_{\theta,m}^{\mathrm{sim,obs}}}.
\label{eq:empirical-sim}
\end{equation}
This empirical distribution represents what the survey would have observed, had the world been governed by parameter $\theta$.

Our discrepancy between observed and simulated data is based on the maximum mean discrepancy (MMD). Let $k:\mathbb R^d\times\mathbb R^d\to\mathbb R$ be a bounded positive-definite kernel, and let $\mathcal H_k$ denote its reproducing kernel Hilbert space. For any probability measure $P$ on $\mathbb R^d$, define the kernel mean embedding
\begin{equation}
\mu_k(P)
=
\int k(y,\cdot)\,P(dy)\in\mathcal H_k.
\end{equation}
The squared maximum mean discrepancy between probability measures $P$ and $Q$ is
\begin{equation}
\mathrm{MMD}_k^2(P,Q)
=
\|\mu_k(P)-\mu_k(Q)\|_{\mathcal H_k}^2.
\end{equation}
Equivalently,
\begin{align}
\mathrm{MMD}_k^2(P,Q)
&=
\mathbb E_{Y,Y'\sim P}[k(Y,Y')]
+
\mathbb E_{W,W'\sim Q}[k(W,W')]
-
2\mathbb E_{Y\sim P, W\sim Q}[k(Y,W)].
\label{eq:mmd-expanded}
\end{align}
When $k$ is characteristic, $\mathrm{MMD}_k(P,Q)=0$ if and only if $P=Q$, so the discrepancy metrizes equality of distributions rather than agreement only in low-order moments.

For each simulated batch at parameter value $\theta$, we define the empirical loss
\begin{equation}
\Delta_{n,M}(\theta)
:=
\mathrm{MMD}_k^2\!\bigl(\widehat P_{\mathrm{obs}},\widehat P_{\theta,M}^{\mathrm{sim,obs}}\bigr).
\label{eq:empirical-loss}
\end{equation}
This loss is small when the simulated observed catalog resembles the actual observed catalog in the geometry induced by the kernel $k$. Unlike an auxiliary-regression summary, $\Delta_{n,M}(\theta)$ compares the empirical distributions directly and therefore retains substantially richer information about the joint structure of the observed variables.

To obtain a likelihood-free posterior surrogate, we draw
\begin{equation}
\theta_1,\dots,\theta_{n_\theta}\stackrel{\mathrm{iid}}{\sim}p(\theta),
\end{equation}
simulate an observed batch of size $M$ for each $\theta_j$, compute the corresponding discrepancy $\Delta_{n,M}(\theta_j)$, and assign kernel weights
\begin{equation}
\widetilde w_j
=
\exp\!\left\{-\frac{\Delta_{n,M}(\theta_j)}{2\tau^2}\right\},
\qquad
w_j
=
\frac{\widetilde w_j}{\sum_{k=1}^{n_\theta}\widetilde w_k},
\end{equation}
where $\tau>0$ is a temperature or bandwidth parameter. The resulting self-normalized empirical measure
\begin{equation}
\Pi_{n_\theta}^{\mathrm{MMD}}(d\theta)
=
\sum_{j=1}^{n_\theta} w_j\,\delta_{\theta_j}(d\theta)
\label{eq:pseudopost}
\end{equation}
serves as a pseudo-posterior over $\theta$. Parameters receiving high weight are those for which the simulator, after noise and selection, generates catalogs that are close in MMD to the observed sample. Since no tractable likelihood is required, the method lies within the broader class of simulation-based or likelihood-free inference procedures, but the loss is defined directly at the level of empirical distributions rather than through a hand-crafted low-dimensional summary.

For population-level arguments, it is useful to introduce the deterministic discrepancy
\begin{equation}
\Delta(\theta)
:=
\mathrm{MMD}_k^2\!\bigl(P_{\mathrm{obs}},P_\theta^{\mathrm{obs}}\bigr),
\label{eq:population-loss}
\end{equation}
together with the idealized population weight
\begin{equation}
L_\infty(\theta;\tau)
=
\exp\!\left\{-\frac{\Delta(\theta)}{2\tau^2}\right\}.
\label{eq:population-weight}
\end{equation}
The corresponding population pseudo-posterior is
\begin{equation}
\pi_\tau(d\theta)
\propto
p(\theta)L_\infty(\theta;\tau)\,d\theta
=
p(\theta)\exp\!\left\{-\frac{\Delta(\theta)}{2\tau^2}\right\}d\theta.
\label{eq:population-pseudo}
\end{equation}
This expression makes clear that the method favors parameter values whose induced observed-data laws lie close to $P_{\mathrm{obs}}$ in maximum mean discrepancy. As $\tau\downarrow 0$, the measure $\pi_\tau$ increasingly concentrates on minimizers of $\Delta(\theta)$, namely on those parameters whose forward-simulated observed distributions most closely reproduce the observed catalog.

It is convenient to define the MMD-identified set
\begin{equation}
\Theta^\dagger
:=
\arg\min_{\theta\in\Theta}\Delta(\theta)
=
\arg\min_{\theta\in\Theta}
\mathrm{MMD}_k^2\!\bigl(P_{\mathrm{obs}},P_\theta^{\mathrm{obs}}\bigr).
\label{eq:identified-set}
\end{equation}
If the simulator is correctly specified in the sense that there exists $\theta^\star\in\Theta$ with
\begin{equation}
P_{\theta^\star}^{\mathrm{obs}}=P_{\mathrm{obs}},
\end{equation}
then $\Delta(\theta^\star)=0$, so $\theta^\star\in\Theta^\dagger$. If $k$ is characteristic, then zero discrepancy implies equality of observed-data laws. This gives the method a much stronger identification target than approaches based on a single moment condition or an auxiliary regression summary. Nevertheless, $\Theta^\dagger$ may still be set-valued if multiple parameter values induce the same observed distribution after the observational pipeline has been applied. Thus the framework naturally accommodates both point identification and partial identification.

The cluster-cosmology setting provides a useful example. Suppose the latent population consists of galaxy clusters generated under cosmological and astrophysical parameters $\theta$ \citep{allen2011cosmological}, where $\theta$ may include quantities such as $\Omega_m$, $\sigma_8$, parameters controlling the halo mass function, and nuisance parameters governing baryonic or selection effects. At the latent level, the simulator produces a population of clusters with true masses, redshifts, temperatures, richnesses, X-ray luminosities, lensing observables, or Sunyaev--Zel'dovich signals. A real survey, however, does not record this latent population directly. Instead, it measures noisy proxies, applies instrumental response, imposes detection thresholds, and retains only objects satisfying catalog selection criteria. The vector $y^{\mathrm{obs}}\in\mathbb R^d$ may therefore contain only the observed catalog-level quantities, for example a stack of measured richness, redshift, X-ray temperature, and signal-to-noise summaries. In this setting, it would be inappropriate to compare the observed catalog directly to the latent simulator output. The scientifically correct comparison is between the real catalog and a simulated catalog obtained by running the latent cluster population through the same survey model. The MMD loss in \eqref{eq:empirical-loss} does exactly this: it quantifies how close the simulated observed cluster catalog is to the actual observed catalog, while retaining sensitivity to the full joint distribution of the recorded quantities.

This perspective also clarifies the distinction between the present method and auxiliary-model approaches such as indirect inference or regression-summary matching. In those methods, observed and simulated datasets are each compressed into a low-dimensional statistic, and inference proceeds by comparing those summaries. Here, by contrast, the discrepancy is defined directly on the empirical distributions of the observed vectors. 

Finally, once the weighted sample $\{(\theta_j,w_j)\}_{j=1}^{n_\theta}$ has been computed, it can be reused in exactly the same way as a posterior sample. Any posterior of interest may be approximated by weighted averages,
\begin{equation}
\int h(\theta)\,\Pi^{\mathrm{MMD}}(d\theta)
\approx
\sum_{j=1}^{n_\theta} w_j h(\theta_j),
\end{equation}
and posterior predictive quantities may be evaluated by forward simulation from the weighted parameter draws. Thus the calibration step need only be performed once, after which the weighted ensemble provides a flexible surrogate posterior for downstream uncertainty quantification and prediction.

\section{Proofs of Theoretical Guarantees}

\begin{proof}[Proof of Theorem~\ref{thm:mc}.]
Write
\begin{equation}
\widetilde w(\theta) = \exp\!\left\{-\frac{\Delta_{n,M}(\theta)}{2\tau^2}\right\}.
\end{equation}
Then
\begin{equation}
\sum_{j=1}^{n_\theta} w_j h(\theta_j) = \frac{n_\theta^{-1}\sum_{j=1}^{n_\theta} h(\theta_j)\widetilde w(\theta_j)}
     {n_\theta^{-1}\sum_{j=1}^{n_\theta} \widetilde w(\theta_j)}.
\end{equation}
Since $0<\widetilde w(\theta)\le 1$, both $h(\theta)\widetilde w(\theta)$ and $\widetilde w(\theta)$ are integrable under the prior whenever $h$ is integrable under $\pi_{n,M,\tau}$. By the strong law of large numbers,
\begin{equation}
\frac{1}{n_\theta}\sum_{j=1}^{n_\theta} h(\theta_j)\widetilde w(\theta_j)
 \xrightarrow{\mathrm{a.s.}} 
\int_\Theta h(\theta)\widetilde w(\theta)p(\theta)\,d\theta
\end{equation}
and
\begin{equation}
\frac1{n_\theta}\sum_{j=1}^{n_\theta} \widetilde w(\theta_j)
 \xrightarrow{\mathrm{a.s.}} 
\int_\Theta \widetilde w(\theta)p(\theta)\,d\theta.
\end{equation}
The denominator limit is strictly positive because $p(\theta)>0$ and $\widetilde w(\theta)>0$ on the compact set $\Theta$. Hence the ratio converges almost surely to
\begin{equation}
\frac{\int_\Theta h(\theta)\widetilde w(\theta)p(\theta)\,d\theta}
{\int_\Theta \widetilde w(\theta)p(\theta)\,d\theta} = \Phi_{n,M,\tau}(h),
\end{equation}
which proves the claim.
\end{proof}

\begin{remark} If $\Delta_{n,M}$ is taken to be the unbiased U-statistic estimator of $\mathrm{MMD}^2$, then $\Delta_{n,M}(\theta)$ may be negative in finite samples. In that case the empirical weight
\begin{equation}
L_{n,M,\tau}(\theta)
=
\exp\!\left\{-\frac{\Delta_{n,M}(\theta)}{2\tau^2}\right\}
\end{equation}
is no longer bounded by $1$. Consequently, finite-sample Monte Carlo consistency requires explicit integrability assumptions on $L_{n,M,\tau}$ and $hL_{n,M,\tau}$ under the prior. The fixed-$\tau$ convergence and posterior concentration results remain valid provided the empirical discrepancy converges uniformly to the nonnegative population loss, and in the shrinking-temperature regime one assumes
\begin{equation}
\sup_{\theta\in\Theta}|\Delta_{n,M}(\theta)-\Delta(\theta)|=\mathcal O_P(\tau_{n,M}^2).
\end{equation}
Alternatively, one may replace $\Delta_{n,M}$ by the nonnegative estimator $\max\{\Delta_{n,M},0\}$ or by the biased V-statistic estimator, which simplifies the weighting arguments.
\end{remark}

\begin{proof}[Proof of Theorem~\ref{thm:fixedtau}.]
Let
\begin{equation}
L_{n,M,\tau}(\theta) 
=
\exp\!\left\{-\frac{\Delta_{n,M}(\theta)}{2\tau^2}\right\},
\qquad
L_\tau(\theta) 
=
\exp\!\left\{-\frac{\Delta(\theta)}{2\tau^2}\right\}.
\end{equation}
Since $\Delta_{n,M}(\theta),\Delta(\theta)\ge 0$ and the map $x\mapsto e^{-x/(2\tau^2)}$ is Lipschitz on $[0,\infty)$ with constant $(2\tau^2)^{-1}$,
\begin{equation}
\sup_{\theta\in\Theta}|L_{n,M,\tau}(\theta)-L_\tau(\theta)|
\le
\frac{1}{2\tau^2}
\sup_{\theta\in\Theta}|\Delta_{n,M}(\theta)-\Delta(\theta)|
\xrightarrow{P}0.
\end{equation}
Now
\begin{equation}
\left|
\int_\Theta h(\theta)p(\theta)L_{n,M,\tau}(\theta)\,d\theta
-
\int_\Theta h(\theta)p(\theta)L_\tau(\theta)\,d\theta
\right|
\end{equation}
\begin{equation}
\le
\int_\Theta |h(\theta)|p(\theta)\,
|L_{n,M,\tau}(\theta)-L_\tau(\theta)|\,d\theta
\le
\|h\|_\infty
\sup_{\theta\in\Theta}|L_{n,M,\tau}(\theta)-L_\tau(\theta)|
\int_\Theta p(\theta)\,d\theta,
\end{equation}
which converges to $0$ in probability. Similarly,
\begin{equation}
\left|
\int_\Theta p(\theta)L_{n,M,\tau}(\theta)\,d\theta
-
\int_\Theta p(\theta)L_\tau(\theta)\,d\theta
\right|
\le
\sup_{\theta\in\Theta}|L_{n,M,\tau}(\theta)-L_\tau(\theta)|
\int_\Theta p(\theta)\,d\theta
\xrightarrow{P}0.
\end{equation}
Finally,
\begin{equation}
\int_\Theta p(\theta)L_\tau(\theta)\,d\theta>0
\end{equation}
because $p(\theta)\ge 0$, $p$ is not identically zero, and $L_\tau(\theta)>0$ for all $\theta$. Therefore Slutsky's theorem implies
\begin{equation}
\frac{\int_\Theta h(\theta)p(\theta)L_{n,M,\tau}(\theta)\,d\theta}
{\int_\Theta p(\theta)L_{n,M,\tau}(\theta)\,d\theta}
\xrightarrow{P}
\frac{\int_\Theta h(\theta)p(\theta)L_\tau(\theta)\,d\theta}
{\int_\Theta p(\theta)L_\tau(\theta)\,d\theta},
\end{equation}
which is the desired conclusion.
\end{proof}

\begin{proof}[Proof of Theorem~\ref{thm:consistency}.]
Recall that
\begin{equation}
\Theta^\dagger
=
\arg\min_{\theta\in\Theta}\Delta(\theta),
\qquad
m^\star:=\inf_{\theta\in\Theta}\Delta(\theta).
\end{equation}
Under Assumptions~\ref{ass:param}--\ref{ass:uniform}, the map $\theta\mapsto\Delta(\theta)$ is continuous on the compact set $\Theta$, so $\Theta^\dagger$ is nonempty and compact, and the infimum $m^\star$ is attained.

Fix an open neighborhood $U\supset\Theta^\dagger$. By Assumption~\ref{ass:ident}, there exists $\eta=\eta(U)>0$ such that
\begin{equation}
\inf_{\theta\notin U}\Delta(\theta)\ge m^\star+\eta.
\label{eq:gap-proof}
\end{equation}
Because $\Theta^\dagger\subset U$ and $\Theta^\dagger\neq\varnothing$, we may choose $\theta_U\in\Theta^\dagger\subset U$. Then
\begin{equation}
\Delta(\theta_U)=m^\star.
\end{equation}
Since $U$ is open and $\Delta$ is continuous, there exists a nonempty open neighborhood $V\subset U$ of $\theta_U$ such that
\begin{equation}
\sup_{\theta\in V}\Delta(\theta)\le m^\star+\frac{\eta}{2}.
\label{eq:V-bound}
\end{equation}

Because $V$ is a nonempty open subset of $\Theta\subset\mathbb R^p$, it has positive Lebesgue measure. Since the prior density $p$ is continuous and strictly positive on $\Theta$, it follows that
\begin{equation}
c_V:=\int_V p(\theta)\,d\theta>0.
\label{eq:cV}
\end{equation}

Now, define the uniform stochastic error
\begin{equation}
\varepsilon_{n,M}
:=
\sup_{\theta\in\Theta}\bigl|\Delta_{n,M}(\theta)-\Delta(\theta)\bigr|.
\end{equation}
By Assumption~\ref{ass:rate},
\begin{equation}
\varepsilon_{n,M}=o_P(\tau_{n,M}^2).
\end{equation}
Consider the event
\begin{equation}
A_{n,M} = \left\{ \varepsilon_{n,M}\le \frac{\eta}{8}\tau_{n,M}^2 \right\}.
\label{eq:AnM}
\end{equation}
Then
\begin{equation}
\Pr(A_{n,M})\longrightarrow 1.
\end{equation}

On the event $A_{n,M}$, for every $\theta\in V$, combining Equation~\eqref{eq:V-bound} with the definition of $\varepsilon_{n,M}$ gives
\begin{align}
\Delta_{n,M}(\theta)
&\le \Delta(\theta)+\varepsilon_{n,M} \nonumber\\
&\le m^\star+\frac{\eta}{2}+\frac{\eta}{8}\tau_{n,M}^2.
\label{eq:upper-V}
\end{align}
Likewise, for every $\theta\notin U$, combining Equation~\eqref{eq:gap-proof} with the definition of $\varepsilon_{n,M}$ gives
\begin{align}
\Delta_{n,M}(\theta)
&\ge \Delta(\theta)-\varepsilon_{n,M} \nonumber\\
&\ge m^\star+\eta-\frac{\eta}{8}\tau_{n,M}^2.
\label{eq:lower-outside}
\end{align}

Therefore, on $A_{n,M}$,
\begin{align}
\sup_{\theta\notin U}\exp\!\left\{-\frac{\Delta_{n,M}(\theta)}{2\tau_{n,M}^2}\right\}
&\le
\exp\!\left\{-\frac{m^\star+\eta-\eta\tau_{n,M}^2/8}{2\tau_{n,M}^2}\right\},
\label{eq:num-bound}
\\
\inf_{\theta\in V}\exp\!\left\{-\frac{\Delta_{n,M}(\theta)}{2\tau_{n,M}^2}\right\}
&\ge
\exp\!\left\{-\frac{m^\star+\eta/2+\eta\tau_{n,M}^2/8}{2\tau_{n,M}^2}\right\}.
\label{eq:den-bound}
\end{align}

Using these bounds, we obtain on $A_{n,M}$,
\begin{align*}
\pi_{n,M,\tau_{n,M}}(\Theta\setminus U)
&=
\frac{\int_{\Theta\setminus U} p(\theta)\exp\{-\Delta_{n,M}(\theta)/(2\tau_{n,M}^2)\}\,d\theta}
{\int_{\Theta} p(\theta)\exp\{-\Delta_{n,M}(\theta)/(2\tau_{n,M}^2)\}\,d\theta}
\\
&\le
\frac{\int_{\Theta\setminus U} p(\theta)\exp\{-\Delta_{n,M}(\theta)/(2\tau_{n,M}^2)\}\,d\theta}
{\int_{V} p(\theta)\exp\{-\Delta_{n,M}(\theta)/(2\tau_{n,M}^2)\}\,d\theta}
\\
&\le
\frac{\exp\!\left\{-\frac{m^\star+\eta-\eta\tau_{n,M}^2/8}{2\tau_{n,M}^2}\right\}
\int_{\Theta\setminus U} p(\theta)\,d\theta}
{\exp\!\left\{-\frac{m^\star+\eta/2+\eta\tau_{n,M}^2/8}{2\tau_{n,M}^2}\right\}
\int_V p(\theta)\,d\theta}.
\end{align*}
Since $\int_{\Theta\setminus U}p(\theta)\,d\theta\le 1$ and $\int_V p(\theta)\,d\theta=c_V$, this yields
\begin{align}
\pi_{n,M,\tau_{n,M}}(\Theta\setminus U)
&\le
\frac{1}{c_V}
\exp\!\left\{
-\frac{\eta/2-(\eta/4)\tau_{n,M}^2}{2\tau_{n,M}^2}
\right\}
\nonumber\\
&=
\frac{1}{c_V}
\exp\!\left\{
-\frac{\eta}{4\tau_{n,M}^2}+\frac{\eta}{8}
\right\}.
\label{eq:final-bound}
\end{align}
Because $\tau_{n,M}\downarrow0$, the right-hand side converges to $0$. Hence
\begin{equation}
\pi_{n,M,\tau_{n,M}}(\Theta\setminus U)\xrightarrow{P}0,
\end{equation}
and therefore
\begin{equation}
\pi_{n,M,\tau_{n,M}}(U)\xrightarrow{P}1,
\end{equation}
which proves the claim in Equation~\eqref{eq:post-consistency}.

For the equivalent formulation, let $F\subset\Theta$ be closed with $F\cap\Theta^\dagger=\varnothing$. Then $U:=\Theta\setminus F$ is open and contains $\Theta^\dagger$. Applying Equation~\eqref{eq:post-consistency} to this set $U$ gives
\begin{equation}
\pi_{n,M,\tau_{n,M}}(F)
=
1-\pi_{n,M,\tau_{n,M}}(U)
\xrightarrow{P}0,
\end{equation}
which proves the claim in Equation~\eqref{eq:closed-consistency}.
\end{proof}

\begin{proof}[Proof of Proposition~\ref{prop:zero}.]
By Assumption~\ref{ass:correct}, the observed distribution is correctly specified at $\theta^\star$, that is,
\[
P_{\theta^\star}^{\mathrm{obs}} = P_{\mathrm{obs}}.
\]
Therefore,
\[
\Delta(\theta^\star) = \mathrm{MMD}_k^2\!\left(P_{\mathrm{obs}}, P_{\theta^\star}^{\mathrm{obs}}\right) = \mathrm{MMD}_k^2\!\left(P_{\mathrm{obs}}, P_{\mathrm{obs}}\right).
\]
Now, by the defining properties of MMD under noisy data (see Theorem~3.9 in \cite{vashistha2026convolutional}), we have
\[
\Delta(\theta^\star)=\mathrm{MMD}_k^2(P_{\mathrm{obs}},P_{\mathrm{obs}})=0.
\]

On the other hand, by Assumption~\ref{ass:kernel}, $\mathrm{MMD}_k^2(\cdot,\cdot)$ is well defined and nonnegative, so
\[
\Delta(\theta)\ge 0
\qquad
\text{for all } \theta\in\Theta.
\]
Since $\Delta(\theta^\star)=0$, it follows that $\inf_{\theta\in\Theta}\Delta(\theta)=0$.

Finally, by definition, $\Theta^\dagger=\arg\min_{\theta\in\Theta}\Delta(\theta)$.
Because $\theta^\star\in\Theta$ and $\Delta(\theta^\star)=0=\inf_{\theta\in\Theta}\Delta(\theta)$, we conclude that $\theta^\star\in\Theta^\dagger$.
\end{proof}

\begin{proof}[Proof of Proposition~\ref{prop:point}.]
By Proposition~\ref{prop:zero}, we have
\[
\inf_{\theta\in\Theta}\Delta(\theta)=0
\qquad\text{and}\qquad
\theta^\star\in\Theta^\dagger.
\]
In particular,
\[
\Delta(\theta^\star)=0.
\]

Now let $\theta\in\Theta^\dagger$. Since $\Theta^\dagger=\arg\min_{\vartheta\in\Theta}\Delta(\vartheta)$ and the minimum value is $0$, it follows that
\[
\Delta(\theta)=0.
\]
By definition of $\Delta$,
\[
\Delta(\theta)
=
\mathrm{MMD}_k^2\!\left(P_{\mathrm{obs}},P_\theta^{\mathrm{obs}}\right)=0.
\]
Because $k$ is characteristic, $\mathrm{MMD}_k(P,Q)=0$ implies $P=Q$. Hence, $P_\theta^{\mathrm{obs}}=P_{\mathrm{obs}}$. Using Assumption~\ref{ass:correct}, we also have $P_{\theta^\star}^{\mathrm{obs}}=P_{\mathrm{obs}}$. Therefore, $P_\theta^{\mathrm{obs}}=P_{\theta^\star}^{\mathrm{obs}}$. 

Since the map $\theta\mapsto P_\theta^{\mathrm{obs}}$ is injective on $\Theta$, this implies $\theta=\theta^\star$. Thus, every $\theta\in\Theta^\dagger$ equals $\theta^\star$, so $\Theta^\dagger=\{\theta^\star\}$.

For the concentration statement, let $U$ be any open neighborhood of $\theta^\star$. Since $\Theta^\dagger=\{\theta^\star\}$, we have $U\supset \Theta^\dagger$. Theorem~\ref{thm:consistency} therefore yields
\[
\pi_{n,M,\tau_{n,M}}(U)\xrightarrow[n,M\to\infty]{P}1.
\]
This completes the proof.
\end{proof}

\begin{proof}[Proof of Corollary~
\ref{cor:two-stage}.]
We prove the fixed-temperature and vanishing-temperature statements separately.

For the first claim, fix $(n,M,\tau)$ with $\tau>0$, and define
\begin{equation}
\widehat\Phi_{n_\theta,n,M,\tau}(h)
:=
\sum_{j=1}^{n_\theta} w_j h(\theta_j)
=
\int_\Theta h(\theta)\,\Pi_{n_\theta}^{\mathrm{MMD}}(d\theta).
\end{equation}
Since $h$ is bounded, it is integrable with respect to $\pi_{n,M,\tau}$. Therefore Theorem~\ref{thm:mc} applies and yields
\begin{equation}
\widehat\Phi_{n_\theta,n,M,\tau}(h)
-
\Phi_{n,M,\tau}(h)
\xrightarrow[n_\theta\to\infty]{\mathrm{a.s.}}0,
\end{equation}
where
\begin{equation}
\Phi_{n,M,\tau}(h)
=
\int_\Theta h(\theta)\,\pi_{n,M,\tau}(d\theta).
\end{equation}
This proves \eqref{eq:two-stage-mc}.

We now prove \eqref{eq:two-stage-fixedtau}. Write
\begin{equation}
\widehat\Phi_{n_\theta,n,M,\tau}(h)-\Phi_\tau(h)
=
\Bigl(\widehat\Phi_{n_\theta,n,M,\tau}(h)-\Phi_{n,M,\tau}(h)\Bigr)
+
\Bigl(\Phi_{n,M,\tau}(h)-\Phi_\tau(h)\Bigr),
\end{equation}
where
\begin{equation}
\Phi_\tau(h):=\int_\Theta h(\theta)\,\pi_\tau(d\theta).
\end{equation}
For the second term, Theorem~\ref{thm:fixedtau} gives
\begin{equation}
\Phi_{n,M,\tau}(h)-\Phi_\tau(h)\xrightarrow[n,M\to\infty]{P}0.
\end{equation}
For the first term, Theorem~\ref{thm:mc} gives almost sure convergence to zero as $n_\theta\to\infty$ for each fixed $(n,M,\tau)$. Hence, along any sequence with $n_\theta\to\infty$, the first term converges to zero in probability. Therefore, by Slutsky's theorem,
\begin{equation}
\widehat\Phi_{n_\theta,n,M,\tau}(h)-\Phi_\tau(h)
\xrightarrow[n,M\to\infty]{P}0,
\end{equation}
which proves \eqref{eq:two-stage-fixedtau}.

We turn to the concentration statement \eqref{eq:two-stage-concentration}. Let $U\supset\Theta^\dagger$ be open. Since $\Pi_{n_\theta}^{\mathrm{MMD}}(U)=\sum_{j=1}^{n_\theta}w_j\mathbf 1_U(\theta_j)$, we may apply Theorem~\ref{thm:mc} with $h=\mathbf 1_U$. Because $0\le \mathbf 1_U\le 1$, the integrability condition is immediate. Thus, for each fixed $(n,M,\tau_{n,M})$,
\begin{equation}
\Pi_{n_\theta}^{\mathrm{MMD}}(U)-\pi_{n,M,\tau_{n,M}}(U)
\xrightarrow[n_\theta\to\infty]{\mathrm{a.s.}}0.
\label{eq:indicator-mc}
\end{equation}
On the other hand, by Theorem~\ref{thm:consistency},
\begin{equation}
\pi_{n,M,\tau_{n,M}}(U)\xrightarrow[n,M\to\infty]{P}1.
\label{eq:empirical-post-concentration}
\end{equation}
Therefore,
\begin{equation}
\Pi_{n_\theta}^{\mathrm{MMD}}(U)-1
=
\Bigl(\Pi_{n_\theta}^{\mathrm{MMD}}(U)-\pi_{n,M,\tau_{n,M}}(U)\Bigr)
+
\Bigl(\pi_{n,M,\tau_{n,M}}(U)-1\Bigr).
\end{equation}
The first term converges to zero in probability because of \eqref{eq:indicator-mc}, and the second term converges to zero in probability because of \eqref{eq:empirical-post-concentration}. Another application of Slutsky's theorem yields
\begin{equation}
\Pi_{n_\theta}^{\mathrm{MMD}}(U)\xrightarrow[n,M\to\infty]{P}1,
\end{equation}
which proves \eqref{eq:two-stage-concentration}.
\end{proof}

\section{Experiment with EIV in Linear Setting}
\label{app:EIV_linear}

\paragraph{Experimental Setup and Benchmarking Protocol.} We evaluate the proposed MMD-based likelihood-free inference framework in a controlled simulation setting designed to closely mirror the problem formulation introduced in Section~\ref{sec:setup}. In particular, the experiment explicitly incorporates measurement error, heteroscedastic noise, and non-Gaussian latent structure, so that inference must operate on observed data distributions rather than latent quantities. The goal is to assess both the accuracy of parameter recovery and the calibration of uncertainty quantification under realistic observational distortions.

\paragraph{Data-Generating Process.} For each parameter setting $\theta = (\alpha, \beta, \sigma)$, we generate a latent population $\{(x_i^{\mathrm{true}}, y_i^{\mathrm{true}})\}_{i=1}^n$. The latent covariate $x_i^{\mathrm{true}}$ is drawn from a two-component Gaussian mixture,
\begin{equation}
x_i^{\mathrm{true}} \sim \pi \,\mathcal{N}(\mu_1, s_1^2) + (1-\pi)\,\mathcal{N}(\mu_2, s_2^2),
\end{equation}
with $\pi = 0.55$, $(\mu_1, \mu_2) = (-1.2, 1.0)$, and $(s_1, s_2) = (0.95, 0.95)$. This induces a mildly multimodal and overlapping distribution, ensuring that the latent covariate distribution is non-Gaussian and cannot be captured by simple parametric assumptions.

Conditional on $x_i^{\mathrm{true}}$, the latent response is generated via a linear model,
\begin{equation}
y_i^{\mathrm{true}} = \alpha + \beta x_i^{\mathrm{true}} + \epsilon_i,
\end{equation}
where $\epsilon_i$ represents intrinsic scatter. We consider two noise regimes: a Gaussian model $\epsilon_i \sim \mathcal{N}(0, \sigma^2)$ and a Laplace model $\epsilon_i \sim \mathrm{Laplace}(0, \sigma/\sqrt{2})$, so that both cases have variance $\sigma^2$. Throughout, we fix $\alpha = 1.0$, $\beta = 2.0$, and $\sigma = 0.8$ unless otherwise specified.

Observed data are obtained by applying a heteroscedastic measurement process,
\begin{equation}
x_i^{\mathrm{obs}} = x_i^{\mathrm{true}} + \eta_{x,i}, 
\qquad
y_i^{\mathrm{obs}} = y_i^{\mathrm{true}} + \eta_{y,i},
\end{equation}
with $\eta_{x,i} \sim \mathcal{N}(0, \sigma_{x,i}^2)$ and $\eta_{y,i} \sim \mathcal{N}(0, \sigma_{y,i}^2)$. The noise levels vary across observations according to
\begin{equation}
\sigma_{x,i} = 1.5 + 0.85\ \frac{|x_i^{\mathrm{true}}-2|}{\max_j |x_j^{\mathrm{true}}|+ 10^{-8}}, 
\qquad
\sigma_{y,i} = 1.5 + 1.20 \ \frac{|y_i^{\mathrm{true}}|}{\max_j |y_j^{\mathrm{true}}| + 10^{-8}}.
\end{equation}

This construction yields large and strongly heteroscedastic measurement errors in both variables. The resulting observed dataset
\[
\mathcal{D}_{\mathrm{obs}} = \{(x_i^{\mathrm{obs}}, y_i^{\mathrm{obs}}, \sigma_{x,i}, \sigma_{y,i})\}_{i=1}^n,
\]
with $n = 180$, corresponds to draws from the observed-data distribution $P_{\mathrm{obs}}$ defined in Section~\ref{sec:setup}.

\paragraph{Inference Methods.} We compare the proposed MMD-based method to a range of classical and Bayesian approaches for regression with measurement error. All methods are applied directly to the observed dataset $\mathcal{D}_{\mathrm{obs}}$.

The proposed method constructs a pseudo-posterior over $\theta$ by comparing the empirical observed distribution $\widehat{P}_{\mathrm{obs}}$ to simulated observed distributions $\widehat{P}_{\theta,M}^{\mathrm{sim,obs}}$. For each sampled parameter value $\theta_j$, we generate a forward-simulated observed dataset by drawing latent variables from $P_\theta^{\mathrm{true}}$ and applying the observation model $P_\theta^{\mathrm{Err}}$. The discrepancy is measured using the squared maximum mean discrepancy,
\begin{equation}
\Delta_{n,M}(\theta) = \mathrm{MMD}_k^2\bigl(\widehat{P}_{\mathrm{obs}}, \widehat{P}_{\theta,M}^{\mathrm{sim,obs}}\bigr),
\end{equation}
and parameters are assigned weights
\begin{equation}
w_j \propto \exp\left(-\frac{\Delta_{n,M}(\theta_j)}{2\tau^2}\right),
\end{equation}
yielding a weighted empirical approximation to the pseudo-posterior. The kernel bandwidth is selected using the median heuristic.

As a state-of-the-art baseline, we employ LinMix \citep{kelly2007some}, a hierarchical Bayesian regression model that accounts for measurement error in both variables and models latent covariates using a Gaussian mixture. Posterior inference is performed using a Gibbs sampler, yielding samples of $(\alpha, \beta, \sigma^2)$.\footnote{\url{https://github.com/jmeyers314/linmix}}

We further include several classical errors-in-variables estimators. Deming regression minimizes orthogonal residuals under a fixed ratio of error variances, which we estimate using the empirical ratio $\mathbb{E}[\sigma_{y,i}^2]/\mathbb{E}[\sigma_{x,i}^2]$ \citep{deming1943statistical}. BCES (Bivariate Correlated Errors and intrinsic Scatter, \citep{akritas1996linear}) provides a regression estimator that accounts for heteroscedastic measurement errors and intrinsic scatter; we use the BCES(Y$|$X) variant implemented in the Python \texttt{bces} package.\footnote{\url{https://pypi.org/project/bces/}} Orthogonal Distance Regression (ODR, \citep{boggs1990orthogonal}), implemented via SciPy, solves a weighted nonlinear least squares problem that accounts for measurement errors in both variables by minimizing orthogonal distances.

We also consider SIMEX (simulation extrapolation, \citep{cook1994simulation}), which corrects for measurement error by adding artificial noise to the observed covariates, fitting a naive model at multiple noise levels, and extrapolating the resulting estimates to the zero-noise limit. Finally, we include ordinary least squares (OLS) as a baseline that ignores measurement error entirely.

For methods without an explicit posterior distribution, including Deming, BCES, and SIMEX, uncertainty is estimated via bootstrap resampling. For ODR and OLS, uncertainty is approximated using asymptotic Gaussian approximations.

\paragraph{Evaluation Metrics.} We evaluate each method in terms of both point estimation accuracy and uncertainty quantification. For each parameter $\theta \in \{\alpha, \beta, \sigma^2\}$, we compute the bias
\begin{equation}
\mathrm{Bias}(\theta) = \widehat{\theta}_{\mathrm{mean}} - \theta^{\mathrm{true}},
\end{equation}
where $\widehat{\theta}_{\mathrm{mean}}$ denotes the posterior or bootstrap mean estimate. We construct 90\% uncertainty intervals $[\theta_{0.05}, \theta_{0.95}]$ and evaluate coverage via the indicator
\begin{equation}
\mathbb{I}\bigl(\theta^{\mathrm{true}} \in [\theta_{0.05}, \theta_{0.95}]\bigr).
\end{equation}
Coverage is reported as the average of this indicator across repeated experiments. For the intrinsic scatter parameter, we evaluate $\sigma^2$ rather than $\sigma$, since it is the natural scale parameter of the likelihood.

\paragraph{Multi-Realization Protocol.} To assess robustness and calibration, we repeat the experiment over $R = 10$ independent realizations. For each realization $r = 1, \dots, R$, we generate a dataset $\mathcal{D}_{\mathrm{obs}}^{(r)}$, fit all methods, and compute posterior summaries and uncertainty intervals. Results are then aggregated across realizations to estimate mean bias, variability of estimates, and empirical coverage.

\paragraph{Discussion.} This experimental setup isolates several key challenges relevant to likelihood-free inference: non-Gaussian latent structure, heteroscedastic and large measurement error, and potential misspecification of intrinsic noise. Classical methods rely on explicit likelihood assumptions and parametric forms for latent distributions, whereas the proposed MMD-based approach directly compares observed-data distributions. By operating at the level of $P_{\mathrm{obs}}$ rather than latent variables, the method naturally incorporates the full observation pipeline and remains robust to misspecification of latent structure and noise distributions.

\begin{figure}
    \centering
    \includegraphics[width=0.75\linewidth]{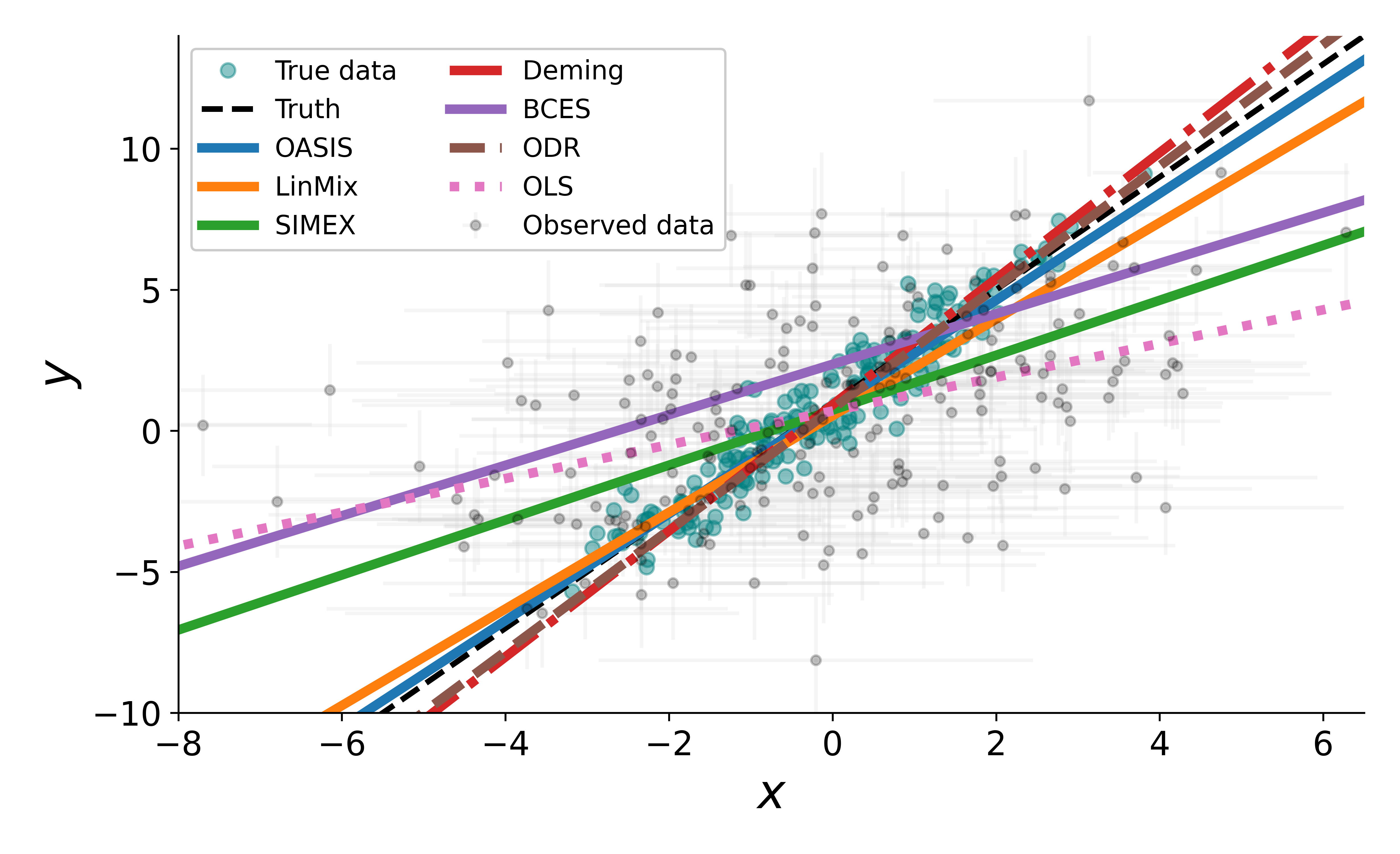}
    \caption{Realization of 180 data points using the setup described in Appendix~\ref{app:EIV_linear}. We compare the estimated linear line fitted against noisy data using benchmark methods.}
    \label{fig:placeholder}
\end{figure}

\begin{figure}
    \centering
    \includegraphics[width=0.98\linewidth]{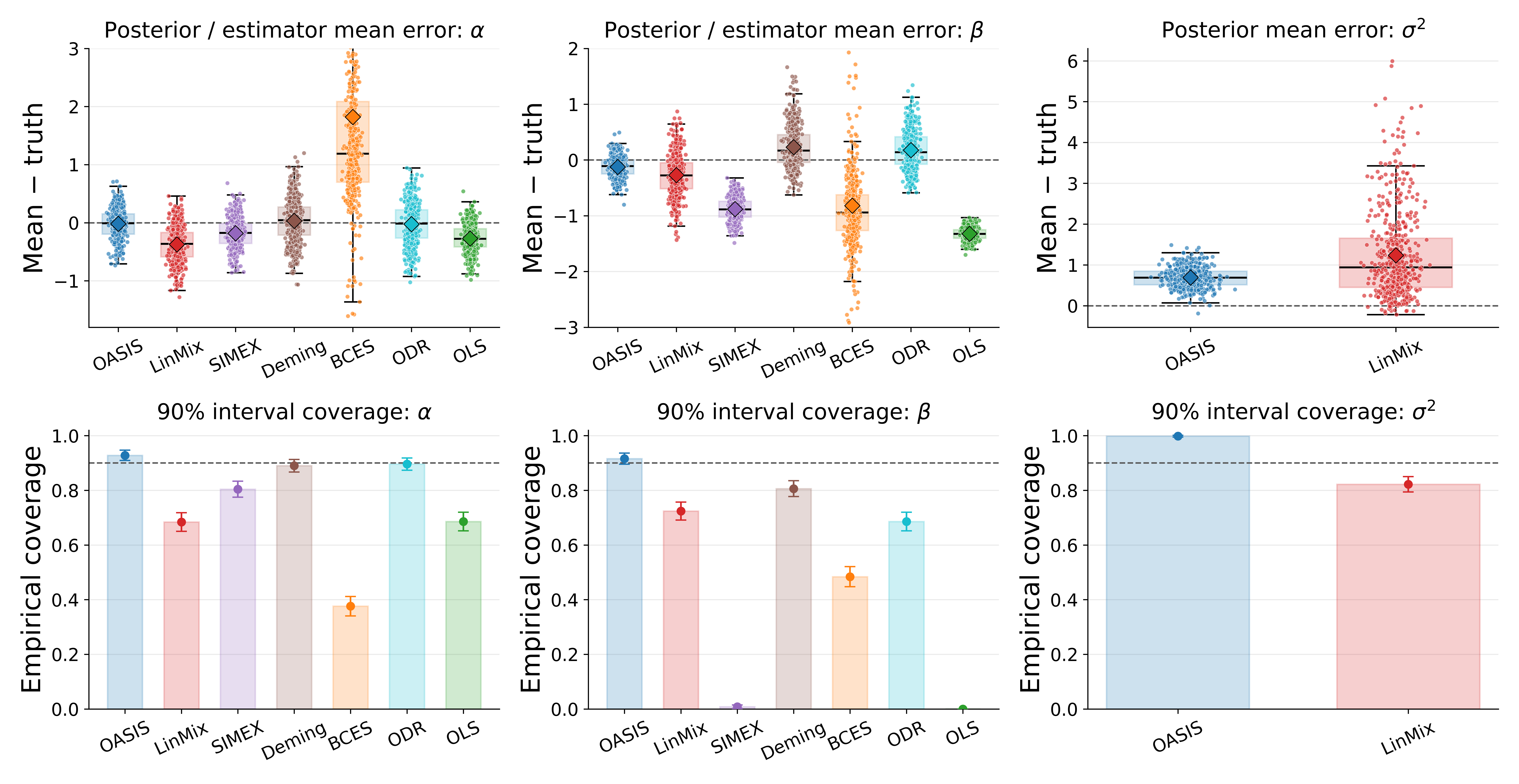} \\
    \includegraphics[width=0.98\linewidth]{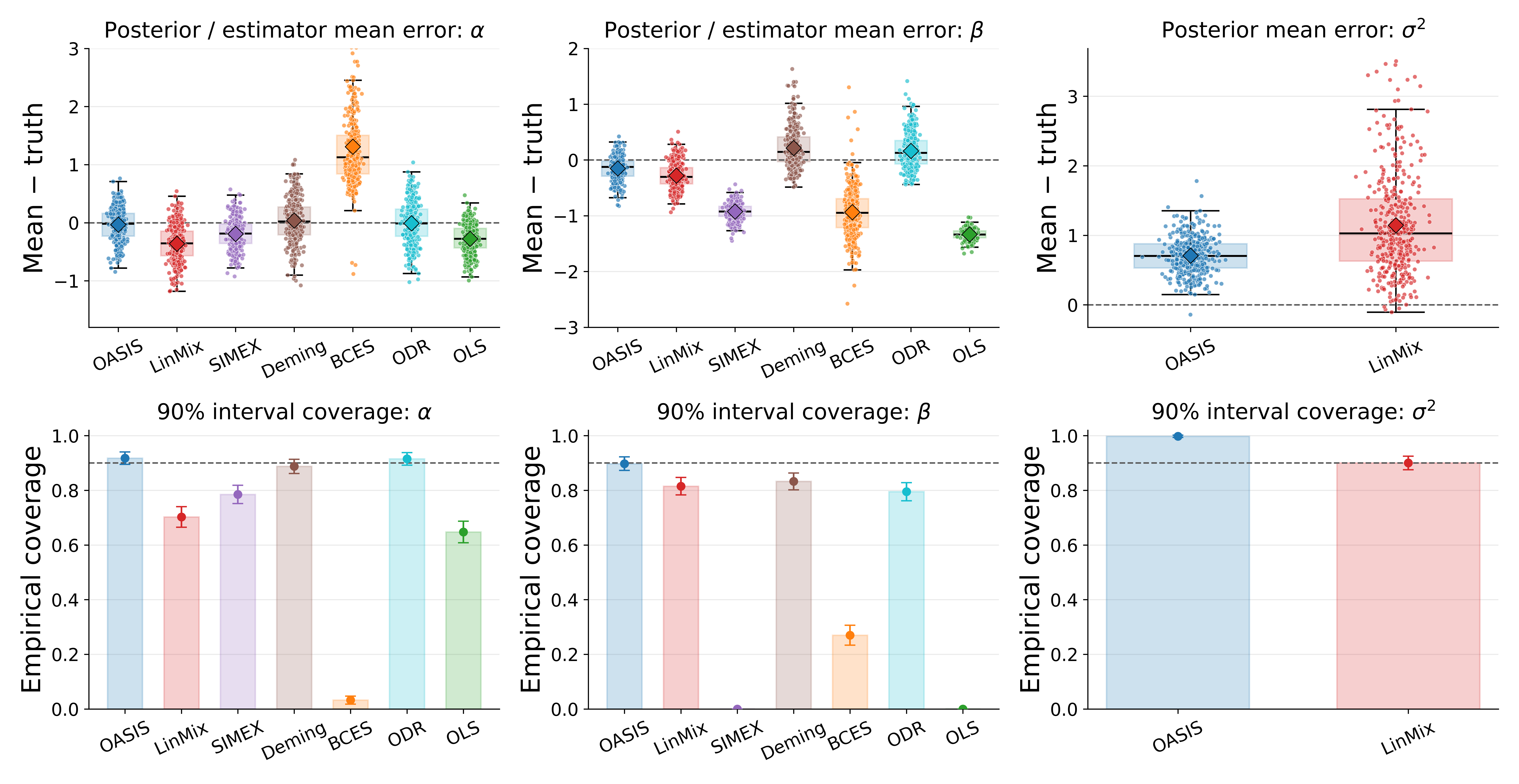}
    \caption{Same as Figure~\ref{fig:EIV_linear_gaussian} but with Laplace (top panels) and Uniform (bottom panels) measurement error.}
    \label{fig:EIV_linear_laplace_uniform}
\end{figure}

\section{Real-world Illustration with Galaxy Cluster Observables}
\label{app:cluster_dgp}

\paragraph{Experimental Setup and Benchmarking Protocol.} 

\paragraph{Data-Generating Process.} We consider four observables: weak lensing mass, optical richness, thermal Sunyaev-Zeldovich effect (SZe) signal-to-noise (SNR), and X-ray count rate $C_{\mathrm{R}}$. In general, cluster cosmology starts from a halo mass function, 

\begin{equation}
    N(m, z) = N_{\mathrm{tot}} P(m,z)
\end{equation}

in which $N_{\mathrm{tot}}$ is the total number of halos within the survey volume, and a redshift range, and $P(m, z)$ is the normalized mass and redshift distribution of dark matter halos. $m$ is the natural logarithm of the mass of the halo, and $z$ is the redshift of the dark matter halo. In general, the halo mass function is a function of the amplitude of density fluctuations $\sigma_8$, total number density $\Omega_{\mathrm{m}}$, the dark energy equation of state $w_0, w_a$, the sum of neutrino masses $\Sigma m_\nu$, and the Hubble constant $H_0$. For the current data, clusters can constrain $\sigma_8$ and $\Omega_{\mathrm m}$ best. In this paper, we will only free $\sigma_8$ and $\Omega_{\mathrm m}$ and keep the rest of the cosmology parameters fixed. 

The underlying true properties of a population of clusters can be described with 

\begin{equation}
    \ln \boldsymbol{y}_i^{\text {true }}=\boldsymbol{\alpha}+ \boldsymbol{\beta} \ln m_i +  \boldsymbol{\gamma} z_i + \boldsymbol{\epsilon}_i, 
\end{equation}

in which $\boldsymbol{y}_i = \{\mathcal{M}^{\mathrm{true}}_i, \lambda_i^{\mathrm{true}}, \xi^{\mathrm{true}}_i, C^{\mathrm{true}}_i \}$ are the true properties of cluster $i$. $\mathcal{M}^{\mathrm{true}}$ is the true weak lensing mass, $\lambda_i^{\mathrm{true}}$ is the true richness, $\xi_i^{\mathrm{true}}$ is the true tSZ SNR, and $C_i^{\mathrm{true}}$ is the true X-ray count rate. The intrinsic scatter and correlation of the true properties are described by 

\begin{equation}
    \boldsymbol{\epsilon}_i \sim \mathcal{N}\left(\mathbf{0}, \boldsymbol{\Sigma}_{\mathrm{int}}\right).
\end{equation}

The observed richness is subject to projection effects that add a heteroscedastic noise that increase the scatter of low-richness clusters \citep{murataConstraintsMassRichnessRelation2018,mylesSpectroscopicQuantificationProjection2021, mylesSpectroscopicCharacterizationRedMaPPer2025}. We parametrized it in the form  

\begin{equation}
    \ln \lambda^{\mathrm{obs}}_i \mid \ln \lambda^{\mathrm{true}}_i \sim \mathcal{N}(\ln \lambda^{\mathrm{true}}_i , a (\ln \lambda^{\mathrm{true}}_i)^2 + b \ln \lambda^{\mathrm{true}}_i ),
\end{equation}

in which $a$ and $b$ are coefficients that control the amplitude of the mass-dependent additional scatter at the low mass end. 

To account for the bias introduced by optimizing over three degrees of freedom, the intrinsic SZe SNR is connected to the observed SZe SNR with 

\begin{equation}
    \xi^{\mathrm{obs}}_i \mid \xi^{\mathrm{true}}_i \sim \mathcal{N}\left(\sqrt{(\xi^{\mathrm{true}})_i^2+3}, 1\right).
\end{equation}

The connection between intrinsic weak lensing mass and the observed weak lensing mass is subject to residual observational effects such as miscentering \citep{zhangDarkEnergySurveyed2019, kellyDarkEnergySurvey2023}, scale-dependent shear calibration bias, and cluster member contamination\citep{vargaDarkEnergySurvey2019}. We model it as a bias plus scatter model 

\begin{equation}
\ln \mathcal{M}^{\mathrm{obs}}_i \mid \ln \mathcal{M}^{\mathrm{true}}_i \sim \mathcal{N}\left(\ln \mathcal{M}^{\mathrm{true}}_i + c (\ln \mathcal{M}^\mathrm{true}_i)^2 + d \ln \mathcal{M}^\mathrm{true}_i, \sigma_{\mathcal{M}}^2 \right), 
\end{equation}

in which $c$ and $d$ describe the amplitude of the mass-dependent bias, and $\sigma_{\mathcal{M}}$ is the observational noise. 

The relation between the true count rate $C^{\mathrm{true}}_i$ and the observed count rate $C^{\mathrm{obs}}_i$ is complicated and requires full physical modeling to capture \citep{sandersHydrostaticChandraXray2018, bulbulSRGEROSITAAllSky2024}. In this work, we use a simple log-normal model 

\begin{equation}
\ln C^{\mathrm{obs}} \mid \ln C^{\mathrm{true}} \sim \mathcal{N}\left(\ln C^{\mathrm{true}}, \sigma_C \right)
\end{equation}

In this experiment, we assume all clusters above a richness threshold are detected in all wavelengths. In the future, we plan to incorporate missing and censored data into the framework.

\paragraph{Benchmarking Protocol.}
We restrict inference to two cosmological parameters, $\theta = (\Omega_m h^2,\,\sigma_8)$, placing flat priors $\Omega_m h^2 \in [0.12, 0.155]$ and $\sigma_8 \in [0.70, 0.90]$ and fixing all scaling-relation nuisance parameters at their fiducial values. Each synthetic observation contains ${\sim}1{,}700$ clusters with five observables per object, $(z,\,\ln\lambda^{\mathrm{obs}},\,\xi,\,\ln\mathcal{M}^{\mathrm{obs}},\,\ln C^{\mathrm{obs}})$, drawn from a $200\,\mathrm{deg}^2$ survey over $z\in[0.1,1.0]$ subject to the richness cut $\lambda^{\mathrm{obs}}\geq20$. The halo mass function is evaluated with the \textsc{MiraTitan} emulator \citep{bocquetMiraTitanUniverseIII2020}. \texttt{OASIS} draws $n_\theta=2{,}000$ proposals from the prior and evaluates $\widehat\Delta(\theta_j)=\widehat{\mathrm{MMD}}^2_k\!\bigl(\widehat P_{\mathrm{obs}},\widehat P_{\theta_j}^{\mathrm{sim}}\bigr)$ using an isotropic RBF kernel $k(\mathbf{x},\mathbf{y})=\exp(-\gamma\|\mathbf{x}-\mathbf{y}\|^2)$ applied jointly to the five-dimensional observable vector, with bandwidth $\gamma=1/(2\hat h^2)$ set by the median heuristic. We set $\tau = 10\,\hat\tau$ where $\hat\tau$ is the heuristic defined in Section~\ref{sec:setup}. We compare against Neural Posterior Estimation (\textbf{NPE}) and Neural Likelihood Estimation (\textbf{NLE}) \citep{cranmer2020frontier} from the \texttt{sbi} library \citep{tejero2020sbi}, each trained on 5{,}000 simulations using a 50-dimensional summary statistic comprising a 2D $(z,\ln\lambda^{\mathrm{obs}})$ histogram, marginal histograms of all five observables, pairwise Pearson correlations, and log-cluster-count features. Performance is assessed over 50 independent realizations, each with a distinct true cosmology drawn uniformly from the prior.

\paragraph{Summary of Results.}
All three methods identify $\sigma_8$ as the primary constraint, as the halo mass function is exponentially sensitive to the amplitude of matter fluctuations. The posterior in $\Omega_m h^2$, however, behaves differently across methods. \texttt{OASIS} shows a nearly prior-width posterior in $\Omega_m h^2$, because the MMD compares two \emph{normalized} empirical measures — dividing by the catalog size $n$ before computing kernel averages — and therefore operates entirely on the per-cluster observable \emph{shape} $(z,\ln\lambda^{\mathrm{obs}},\xi,\ln\mathcal{M}^{\mathrm{obs}},\ln C^{\mathrm{obs}})$, discarding the total cluster count $N_\mathrm{tot}$. Since $N_\mathrm{tot}$ is the primary channel through which $\Omega_m h^2$ enters the observable, this normalization renders the MMD discrepancy virtually insensitive to $\Omega_m h^2$. NPE and NLE, by contrast, include explicit log-cluster-count features in their 50-dimensional summary statistic, recovering partial sensitivity to $\Omega_m h^2$ (See Section~\ref{app:distribution_vs_abundance} for more details). However, this apparent gain comes at the cost of calibration: NPE reaches only $70\%$ and $72\%$ coverage for $\Omega_m h^2$ and $\sigma_8$ respectively, and NLE reaches $74\%$ and $76\%$, well below the nominal $90\%$, indicating that the neural posteriors are overconfident rather than genuinely better constrained. \texttt{OASIS}, meanwhile, achieves near-zero mean bias ($\Delta\Omega_m h^2 = -0.004$, $\Delta\sigma_8 = 0.001$) and conservative-to-nominal coverage ($82\%$ and $98\%$), confirming that its intervals are well-calibrated at the expense of width. The bias of all three methods is comparably small, so the undercoverage of the neural methods reflects posterior overconfidence, not a systematic shift.

\paragraph{Computational Cost.}
\texttt{OASIS} requires no training phase: given an observed catalog, the 2{,}000 forward simulations are embarrassingly parallel and complete in approximately $1.2$ minutes per realization on 128 CPU workers. Over the 50-realization benchmark, this amounts to roughly $60$ minutes of wall time. The cost scales linearly with the number of proposals and is paid anew for each observation, as the method is not amortized. NPE and NLE, by contrast, incur a one-time training cost of approximately $6$ minutes on 5{,}000 simulations; once trained, posterior inference for any new observation reduces to a single neural-network forward pass taking under a second.

\section{Information Content of Distribution Matching vs.\ Abundance Likelihoods}
\label{app:distribution_vs_abundance}

In cluster cosmology, the statistical constraining power arises from two distinct sources of information: the total abundance of detected clusters and the distribution of their observable properties \cite{allen2011cosmological}. The expected number of detected clusters under cosmological parameters $\theta$ is

\begin{equation}
\mu(\theta)
=
\int dz \, dm \,
S(m,z)\,
\frac{dn(m,z\mid\theta)}{dm}
\frac{dV(z\mid\theta)}{dz},
\end{equation}

where $dn/dm$ is the halo mass function, $dV/dz$ is the comoving volume element, and $S(m,z)$ denotes the effective survey selection function. The corresponding likelihood for a catalog containing $N_{\mathrm{obs}}$ detected clusters with observables $\{x_i\}_{i=1}^{N_{\mathrm{obs}}}$ can be written schematically as

\begin{equation}
\mathcal{L}(\theta) = \mathrm{Pois}(N_{\mathrm{obs}}\mid \mu(\theta)) \prod_{i=1}^{N_{\mathrm{obs}}} p(x_i\mid\theta),
\label{eq:cluster_likelihood_factorization}
\end{equation}
where
\begin{equation}
p(x\mid\theta) = \frac{1}{\mu(\theta)} \frac{d\mu(x\mid\theta)}{dx}
\end{equation}

is the normalized distribution of cluster observables. Equation~\eqref{eq:cluster_likelihood_factorization} makes explicit that cosmological information enters through two channels:

\begin{enumerate}
    \item the Poisson abundance term, which constrains the overall normalization of the halo population, and
    \item the normalized observable distribution, which constrains the shape of the detected population in mass, redshift, and observable space.
\end{enumerate}

The NPE and NLE baselines considered in this work use hand-designed summary statistics that explicitly include cluster-count information, marginal histograms, and redshift-richness distributions. Consequently, these methods retain information from both the abundance normalization and the observable distribution. By contrast, the current \texttt{OASIS} implementation compares empirical observable distributions using the Maximum Mean Discrepancy (MMD),

\begin{equation}
\widehat{\Delta}(\theta)
=
\widehat{\mathrm{MMD}}^2_k
\left(
\widehat{P}_{\mathrm{obs}},
\widehat{P}^{\mathrm{sim}}_{\theta}
\right),
\end{equation}

where $\widehat{P}_{\mathrm{obs}}$ and $\widehat{P}^{\mathrm{sim}}_{\theta}$ are normalized empirical distributions. As implemented here, this comparison is primarily sensitive to changes in the shape of the observable population rather than to the absolute number of detected clusters. Information carried exclusively by the abundance normalization is therefore partially discarded.

This distinction is particularly relevant for $\Omega_m h^2$. In the cluster mass and redshift regime probed in this benchmark, the halo mass function depends exponentially on $\sigma_8$, making the shape of the observed population highly informative for that parameter. In contrast, $\Omega_m h^2$ enters more weakly through the growth history, power-spectrum shape, and cosmological volume element \citep{allen2011cosmological}. As a result, much of the constraining power on $\Omega_m h^2$ arises from the overall abundance normalization rather than from subtle changes in the normalized observable distribution alone. Discarding abundance information, therefore, naturally broadens the posterior along this direction.

The broader posterior intervals obtained by \texttt{OASIS} for $\Omega_m h^2$ should therefore not be interpreted as a failure of the method, but rather as a consequence of a deliberately restricted information representation. The goal of our experiment is not to maximize cosmological precision, but instead to demonstrate that \texttt{OASIS} can operate in a realistic cluster-cosmology setting involving correlated multi-wavelength observables, heterogeneous survey overlap, and complex selection effects, while still producing competitive and well-calibrated constraints without requiring hand-crafted summary statistics or amortized neural training.

Future extensions could incorporate abundance information directly into the distributional comparison objective, for example, by augmenting the MMD with count-sensitive terms or by jointly comparing both normalized distributions and total catalog counts. Such extensions would allow \texttt{OASIS} to recover a larger fraction of the cosmological information available in cluster abundance surveys while retaining the flexibility of likelihood-free distribution matching.


\end{document}